\documentclass[peerreview,a4paper,12pt]{IEEEtran}

\usepackage{amsthm}
\usepackage{amsmath}
\usepackage{amsfonts,amssymb,amsmath}
\usepackage{caption}
\usepackage{graphicx}
\usepackage{epsfig}
\usepackage[applemac]{inputenc}
\usepackage{latexsym}

\newtheorem{mydef}{Definition}
\newtheorem{mycor}{Corollary}
\newtheorem{myprop}{Proposition}
\newtheorem{mythe}{Theorem}
\newtheorem{mylem}{Lemma}
\newtheorem{mynot}{Notation}
\newtheorem{myrem}{Remark}

\newtheorem{myconj}{Conjecture}
\newtheorem{myex}{Example}

\begin{document}

\sloppy

\title{Polar Codes for Arbitrary DMCs and Arbitrary MACs} 

\author{Rajai Nasser and Emre Telatar,~\IEEEmembership{Fellow,~IEEE,}\\
School of Computer and Communication Sciences, EPFL\\
Lausanne, Switzerland\\
Email: \{rajai.nasser, emre.telatar\}@epfl.ch
\thanks{This paper was presented in part at the IEEE International Symposium on Information Theory, Istanbul, Turkey, July 2013.}
\thanks{This paper is submitted to the IEEE Transactions on Information Theory.}
}

\maketitle



\begin{abstract}
Polar codes are constructed for arbitrary channels by imposing an arbitrary quasigroup structure on the input alphabet. Just as with ``usual" polar codes, the block error probability under successive cancellation decoding is $o(2^{-N^{1/2-\epsilon}})$, where $N$ is the block length. Encoding and decoding for these codes can be implemented with a complexity of $O(N\log N)$. It is shown that the same technique can be used to construct polar codes for arbitrary multiple access channels (MAC) by using an appropriate Abelian group structure. Although the symmetric sum capacity is achieved by this coding scheme, some points in the symmetric capacity region may not be achieved. In the case where the channel is a combination of linear channels, we provide a necessary and sufficient condition characterizing the channels whose symmetric capacity region is preserved by the polarization process. We also provide a sufficient condition for having a maximal loss in the dominant face.
\end{abstract}


\section{Introduction}

Polar coding, invented by Ar{\i}kan \cite{Arikan}, is the first low complexity coding technique that achieves the capacity of binary-input symmetric memoryless channels. Polar codes rely on a phenomenon called \emph{polarization}, which is the process of converting a set of identical copies of a given single user binary-input channel, into a set of ``almost extremal channels", i.e., either ``almost perfect channels'', or ``almost useless channels". The probability of error of successive cancellation decoding of polar codes was proven to be equal to $o(2^{-N^{1/2-\epsilon}})$ by Ar{\i}kan and Telatar \cite{ArikanTelatar}.

Ar{\i}kan's technique was generalized by \c{S}a\c{s}o\u{g}lu et al. for channels with an input alphabet of prime size \cite{SasogluTelAri}. Generalization to channels with arbitrary input alphabet size is not simple since it was shown in \cite{SasogluTelAri} that if we use a group operation in an Ar{\i}kan-like construction, it is not guaranteed that polarization will happen as usual to ``almost perfect channels" or ``almost useless channels". \c{S}a\c{s}o\u{g}lu \cite{SasS} used a special type of quasigroup operation to ensure polarization.

Park and Barg \cite{ParkBarg} showed that polar codes can be constructed using the group structure $\mathbb{Z}_{2^r}$. Sahebi and Pradhan \cite{SahebiPradhan} showed that polar codes can be constructed using any Abelian group structure. The polarization phenomenon described in \cite{ParkBarg} and \cite{SahebiPradhan} does not happen in the usual sense, indeed, it was previously proven by \c{S}a\c{s}o\u{g}lu et al. that it is not the case. It is shown in \cite{ParkBarg} and \cite{SahebiPradhan} that while it is true that we don't always have polarization to ``almost perfect channels" or ``almost useless channels" if a general Abelian operation is used, we always have polarization to ``almost useful channels" (i.e., channels that are easy to be used for communication). The proofs in \cite{ParkBarg} and \cite{SahebiPradhan} rely mainly on the properties of Battacharyya parameters to derive polarization results. In this paper, we adopt a different approach: we give a direct elementary proof of polarization for the more general case of quasigroups using only elementary information theoretic concepts (namely, entropies and mutual information). The Battacharyya parameter is used here only to derive the rate of polarization.

In the case of multiple access channels (MAC), we find two main results in the literature: (i) \c{S}a\c{s}o\u{g}lu et al. constructed polar codes for the two-user MAC with an input alphabet of prime size \cite{SasogluTelYeh}, (ii) Abbe and Telatar used matroid theory to construct polar codes for the $m$-user MAC with binary input \cite{AbbeTelatar}. The generalization of the results in \cite{AbbeTelatar} to MACs with arbitrary input alphabet size is not trivial even in the case of prime size since there is no known characterization for non-binary matroids. We have shown in \cite{Rajai} that the use of matroid theory is not necessary; we used elementary techniques to construct polar codes for the $m$-user MAC with input alphabet of prime size. In this paper, we will see how we can construct polar codes for an arbitrary MAC where the input alphabet size is allowed to be arbitrary, and possibly different from one user to another.

In our construction, as well as in both constructions in \cite{SasogluTelYeh} and \cite{AbbeTelatar}, the symmetric sum capacity is preserved by the polarization process. However, a part of the symmetric capacity region may be lost in the process. We study this loss in the special case where the channel is a combination of linear channels (this class of channels will be introduced in section 8).

In section 2, we introduce the preliminaries for this paper. We describe the polarization process in section 3. The rate of polarization is studied in section 4. Polar codes for arbitrary single user channels are constructed in section 5. The special case of group structures is discussed in section 6. We construct polar codes for arbitrary MAC in section 7. The problem of loss in the capacity region is studied in section 8.

\section{Preliminaries}

We first recall the definitions for multiple access channels in order to introduce the notation that will be used throughout this paper. Since ordinary channels (one transmitter and one receiver) can be seen as a special case of multiple access channels, we will not provide definitions for ordinary channels.

\subsection{Multiple access channels}

\begin{mydef}
A discrete $m$-user multiple access channel (MAC) is an ($m+2$)-tuple $P=(\mathcal{X}_1,\;\mathcal{X}_2,\;\ldots,\;\mathcal{X}_m,\;\mathcal{Y},\;f_P)$ where $\mathcal{X}_1,\;\ldots,\;\mathcal{X}_m$ are finite sets that are called the \emph{input alphabets} of $P$, $\mathcal{Y}$ is a finite set that is called the \emph{output alphabet} of $P$, and $f_P:\mathcal{X}_1\times\mathcal{X}_2\times\ldots\times\mathcal{X}_m \times \mathcal{Y} \rightarrow [0,1]$  is a function satisfying $\forall(x_1,x_2,\ldots,x_m)\in\mathcal{X}_1\times\mathcal{X}_2\times\ldots\times\mathcal{X}_m,\;\displaystyle\sum_{y\in\mathcal{Y}} f_P(x_1,x_2,\ldots,x_m,y)=1$.
\end{mydef}

\begin{mynot}
We write $P: \mathcal{X}_1\times\mathcal{X}_2\times\ldots\times\mathcal{X}_m \rightarrow \mathcal{Y}$ to denote that $P$ has $m$ users, $\mathcal{X}_1,\;\mathcal{X}_2,\;\ldots,\;\mathcal{X}_m$ as input alphabets, and $\mathcal{Y}$ as output alphabet. We denote $f_P(x_1,x_2,\ldots,x_m,y)$ by $P(y|x_1,x_2,\ldots,x_m)$ which is interpreted as the conditional probability of receiving $y$ at the output, given that $(x_1,x_2,\ldots,x_m)$ is the input.
\end{mynot}

\begin{mydef}
A code $\mathcal{C}$ of block length $N$ and rate vector $(R_1,R_2,\ldots,R_m)$ is an $(m+1)$-tuple $\mathcal{C}=(f_1,f_2,\ldots,f_m,g)$, where $f_k: \mathcal{W}_k=\{1,2,\ldots,e^{NR_k}\}\rightarrow\mathcal{X}_k^N$ is the encoding function of the $k^{th}$ user and $g:\mathcal{Y}^n\rightarrow\mathcal{W}_1\times\mathcal{W}_2\times\ldots\times\mathcal{W}_m$ is the decoding function. We denote $f_k(w)=\big(f_k(w)_1,\ldots,f_k(w)_N\big)$, where $f_k(w)_n$ is the $n^{th}$ component of $f_k(w)$. The average probability of error of the code $\mathcal{C}$ is given by:
\begin{align*}
P_e(\mathcal{C})=\sum_{(w_1,\ldots,w_m)\in\mathcal{W}_1\times\ldots\times\mathcal{W}_m}
\frac{P_e(w_1,\ldots,w_m)}{|\mathcal{W}_1|\times\ldots\times|\mathcal{W}_m|},
\end{align*}
\begin{align*}
P_e(w_1,\ldots,w_m)=\sum_{\substack{(y_1,\ldots,y_N)\in\mathcal{Y}^N\\g(y_1,\ldots,y_N) \neq(w_1,\ldots,w_m)}}\prod_{n=1}^N P\big(y_n|f_1(w_1)_n,\ldots,f_m(w_m)_n\big).
\end{align*}
\end{mydef}

\begin{mydef}
A rate vector $R=(R_1,\ldots,R_m)$ is said to be achievable if there exists a sequence of codes $\mathcal{C}_N$ of rate vector $(R_1-\epsilon_{1,N},R_2-\epsilon_{2,N},\ldots,R_m-\epsilon_{m,N})$ and of block length $N$ such that the sequence $\{P_e(\mathcal{C}_N)\}_N$ and the sequences $\{\epsilon_{i,N}\}_N$ (for all $1\leq i\leq m$) tend to zero as $N$ tends to infinity. The capacity region of the MAC $P$ is the set of all achievable rate vectors.
\end{mydef}

\begin{mydef}
Given a MAC $P$ and a collection of independent random variables $X_1,\ldots,X_m$ taking values in $\mathcal{X}_1,\ldots,\mathcal{X}_m$ respectively, we define the  polymatroid region $\mathcal{J}_{X_1,\ldots,X_m}(P)$ in $\mathbb{R}^m$ by:
\begin{align*}
\mathcal{J}_{X_1,\ldots,X_m}(P):=\big\{ R=(R_1,\ldots,R_m)\in \mathbb{R}^m:0\leq R(S) \leq I_{X_1,\ldots,X_m}[S](P)\;\;\mathrm{for\;all}\;\;S\subset\{1,\ldots,m\}\big\},
\end{align*}
where $R(S):=\displaystyle\sum_{k=1}^{l_S} R_k$, $X(S):=(X_{s_1},\ldots,X_{s_{l_S}})$ for $S=\{s_1,\ldots,s_{l_S}\}$ and $I_{X_1,\ldots,X_m}[S](P):=I(X(S);YX(S^c))$. The mutual information is computed for the probability distribution $P(y|x_1,\ldots,x_m)\emph{\textrm{P}}_{X_1,\ldots,X_m}(x_1,\ldots,x_m)$ on $\mathcal{X}_1\times\ldots\times\mathcal{X}_m\times\mathcal{Y}$.
\end{mydef}

\begin{mythe}
(Theorem 15.3.6 \cite{Cover}) The capacity region of a MAC $P$ is given by the closure of the convex hull of the union of all information theoretic capacity regions of $P$ for all the input distributions, i.e, $\displaystyle \overline{\mathrm{ConvexHull}}\Bigg(\bigcup_{\substack{X_1,\ldots,X_m\\ \mathrm{are\;independent}\\\mathrm{random\;variables\;in}\\\mathcal{X}_1,\ldots,\mathcal{X}_m\;\mathrm{resp.}}}\mathcal{J}_{X_1,\ldots,X_m}(P)\Bigg)$.
\end{mythe}

\begin{mydef}
$I_{X_1,\ldots,X_m}(P):=I_{X_1,\ldots,X_m}[\{1,\ldots,m\}](P)$ is called the \emph{sum capacity} of $P$ for the input distributions $X_1,\ldots,X_m$. It is equal to the maximum value of $R_1+\ldots+R_m$ when $(R_1,\ldots,R_m)$ belongs to the information theoretic capacity region for input distributions $X_1,\ldots,X_m$. The set of points of the information theoretic capacity region satisfying $R_1+\ldots+R_m=I_{X_1,\ldots,X_m}(P)$ is called \emph{the dominant face} of this region.
\end{mydef}

\begin{mynot}
When $X_1,\ldots,X_m$ are independent and uniform random variables in $\mathcal{X}_1,\ldots,\mathcal{X}_m$ respectively, we will simply denote $\mathcal{J}_{X_1,\ldots,X_m}(P)$, $I_{X_1,\ldots,X_m}[S](P)$ and $I_{X_1,\ldots,X_m}(P)$ by $\mathcal{J}(P)$, $I[S](P)$ and $I(P)$ respectively. $\mathcal{J}(P)$ is called the \emph{symmetric capacity region} of $P$, and $I(P)$ is called the \emph{symmetric sum capacity} of $P$.
\end{mynot}

\subsection{Quasigroups}

\begin{mydef}
A quasigroup is a pair $(Q,\ast)$, where $\ast$ is a binary operation on the set $Q$ satisfying the following:
\begin{itemize}
\item For any two elements $a,b\in Q$, there exists a unique element $c\in Q$ such that $a=b\ast c$. We denote this element $c$ by $b\backslash_{\ast} a$.
\item For any two elements $a,b\in Q$, there exists a unique element $d\in Q$ such that $a=d\ast b$. We denote this element $d$ by $a/^{\ast}b$.
\end{itemize}
\end{mydef}

\begin{myrem}
If $(Q,\ast)$ is a quasigroup, then $(Q,/^{\ast})$ and $(Q,\backslash_{\ast})$ are also quasigroups.
\end{myrem}

\begin{mynot}
Let $A$ and $B$ be two subsets of a quasigroup $(Q,\ast)$. We define the set:
$$A\ast B:=\{a\ast b:\; a\in A,b\in B\}.$$
If $A$ and $B$ are non-empty, then $|A\ast B|\geq \max\{|A|,|B|\}$.
\end{mynot}

\begin{mydef}
Let $Q$ be any set. A partition $\mathcal{H}$ of $Q$ is said to be a \emph{balanced partition} if and only if all the elements of $\mathcal{H}$ have the same size. We denote the common size of its elements by $||\mathcal{H}||$. The number of elements in $\mathcal{H}$ is denoted by $|\mathcal{H}|$. Clearly, $|Q|=|\mathcal{H}|\times||\mathcal{H}||$ for such a partition.
\end{mydef}

\begin{mydef}
Let $\mathcal{H}$ be a balanced partition of a set $Q$. We define the \emph{projection onto} $\mathcal{H}$ as the mapping $\emph{Proj}_{\mathcal{H}}:Q\longrightarrow\mathcal{H}$, where $\emph{Proj}_{\mathcal{H}}(x)$ is the unique element $H\in\mathcal{H}$ such that $x\in H$.
\end{mydef}

\begin{mylem}
\label{lemsize}
Let $\mathcal{H}$ be a balanced partition of a quasigroup $(Q,\ast)$. Define: $$\mathcal{H}^*:=\{A\ast B:\; A,B\in\mathcal{H}\}.$$ If $\mathcal{H}^{\ast}$ is a balanced partition, then $||\mathcal{H}^\ast||\geq ||\mathcal{H}||$.
\end{mylem}
\begin{proof}
Let $A,B\in\mathcal{H}$ then $A\ast B\in\mathcal{H}^{\ast}$, we have:
$$||\mathcal{H}^{\ast}||=|A\ast B|\geq \max\{|A|,|B|\}=||\mathcal{H}||.$$
\end{proof}

\begin{mydef}
\label{defbal}
Let $(Q,\ast)$ be a quasigroup. A balanced partition $\mathcal{H}$ of $Q$ is said to be a \emph{stable partition} of \emph{period} $n$ of $(Q,\ast)$ if and only if there exist $n$ different balanced partitions $\mathcal{H}_1$, \ldots, $\mathcal{H}_n$ of $Q$ such that:
\begin{itemize}
\item $\mathcal{H}_1=\mathcal{H}$.
\item $\mathcal{H}_{i+1}=\mathcal{H}_i^{\ast}=\{A\ast B:\; A,B\in\mathcal{H}_i\}$ for all $i\leq n-1$.
\item $\mathcal{H}=\mathcal{H}_n^{\ast}$.
\end{itemize}
It is easy to see that if $\mathcal{H}$ is a stable partition of period $n$, then $||\mathcal{H}_i||=||\mathcal{H}||$ for all $1\leq i\leq n$ (from lemma \ref{lemsize},  we have $||\mathcal{H}||=||\mathcal{H}_1||\leq||\mathcal{H}_2||\leq\ldots\leq ||\mathcal{H}_n||\leq ||\mathcal{H}||$).
\end{mydef}

\begin{myrem}
Stable partitions always exist. Any quasigroup $(Q,\ast)$ admits at least the following two stable partitions of period 1: $\{Q\}$ and $\big\{\{x\}:x\in Q\big\}$, which are called the \emph{trivial stable partitions} of $(Q,\ast)$. It is easy to see that when $|Q|$ is prime, the only stable partitions are the trivial ones.
\end{myrem}

\begin{myex}
Let $Q=\mathbb{Z}_{n}\times\mathbb{Z}_{n}$, define $(x_1,y_1)\ast (x_2,y_2)=(x_1+y_1+x_2+y_2,y_1+y_2)$. For each $j\in\mathbb{Z}_n$ and each $1\leq i\leq n$, define $H_{i,j}=\{(j+(i-1)k,k):\; k\in\mathbb{Z}_n\}$. Let $\mathcal{H}_i=\{H_{i,j}:\; j\in\mathbb{Z}_n\}$ for $1\leq i\leq n$. It is easy to see that $\mathcal{H}_i^{\ast}=\mathcal{H}_{i+1}$ for $1\leq i\leq n-1$ and $\mathcal{H}_n^{\ast}=\mathcal{H}_1$. Therefore, $\mathcal{H}:=\mathcal{H}_1$ is a stable partition of $(Q,\ast)$ whose period is $n$.
\end{myex}
Note that the operation $\ast$ in the last example is not a group operation when $n>1$.

\begin{mylem}
\label{lemarb}
If $\mathcal{H}$ is a stable partition and $A_1$ is an arbitrary element of $\mathcal{H}$, then $\mathcal{H}^\ast=\{A_1\ast A_2:\;A_2\in\mathcal{H}\}$.
\end{mylem}
\begin{proof}
We have:
\begin{align*}
Q &= A_1\ast Q= A_1\ast\Big(\bigcup_{A_2\in\mathcal{H}}A_2\Big) =\bigcup_{A_2\in\mathcal{H}}(A_1\ast A_2).
\end{align*}
Therefore, $\{A_1\ast A_2:\; A_2\in\mathcal{H}\}$ covers $Q$ and is a subset of $\mathcal{H}^\ast$ (which is a partition of $Q$ that does not contain the empty set as an element). We conclude that $\mathcal{H}^{\ast}=\{A_1\ast A_2:\; A_2\in\mathcal{H}\}$.
\end{proof}

\begin{mydef}
For any two partitions $\mathcal{H}_1$ and $\mathcal{H}_2$, we define:
$$\mathcal{H}_1 \wedge \mathcal{H}_2=\{A\cap B:\; A\in\mathcal{H}_1,\; B\in\mathcal{H}_2,\; A\cap B\neq \o\}.$$
\end{mydef}

\begin{mylem}
\label{lemWedge}
If $\mathcal{H}_1$ and $\mathcal{H}_2$ are stable then $\mathcal{H}_1\wedge\mathcal{H}_2$ is also a stable partition of $(Q,\ast)$, and $(\mathcal{H}_1\wedge\mathcal{H}_2)^{\ast}=\mathcal{H}_1^{\ast}\wedge \mathcal{H}_2^{\ast}$.
\end{mylem}
\begin{proof}
Since $\mathcal{H}_1$ and $\mathcal{H}_2$ are two partitions of $Q$, it is easy to see that $\mathcal{H}_1\wedge\mathcal{H}_2$ is also a partition of $Q$. Now let $A_1,A_2\in \mathcal{H}_1$ and $B_1,B_2\in \mathcal{H}_2$. If $A_1\cap B_1\neq \o$ and $A_2\cap B_2\neq \o$, we have:
\begin{equation}
\label{eqref}
(A_1\cap B_1)\ast (A_2\cap B_2)\;\subset\;(A_1\ast A_2)\cap (B_1\ast B_2)\in \mathcal{H}_1^{\ast}\wedge\mathcal{H}_2^{\ast}.
\end{equation}
Let $A_1\in\mathcal{H}_1$ and $B_1\in\mathcal{H}_2$ be chosen such that $|A_1\cap B_1|$ is maximal. Lemma \ref{lemarb} implies that $\mathcal{H}_1^{\ast}=\{A_1\ast A_2:\; A_2\in\mathcal{H}_1\}$ and $\mathcal{H}_2^{\ast}=\{B_1\ast B_2:\; B_2\in\mathcal{H}_1\}$. Therefore,
$$|Q| =\sum_{\substack{(A_2,B_2)\in \mathcal{H}_1\times\mathcal{H}_2}}|(A_1\ast A_2)\cap (B_1\ast B_2)|,$$
which implies that
\begin{align}
\label{eqe1}
|Q| &\geq \sum_{\substack{(A_2,B_2)\in \mathcal{H}_1\times\mathcal{H}_2\\ A_2\cap B_2\neq \o}}|(A_1\ast A_2)\cap (B_1\ast B_2)|\\
&\geq \label{eqe2}
\sum_{\substack{(A_2,B_2)\in \mathcal{H}_1\times\mathcal{H}_2\\ A_2\cap B_2\neq \o}}|(A_1\cap B_1)\ast (A_2\cap B_2)|,
\end{align}
where \eqref{eqe2} follows from \eqref{eqref}. Now if $A_2\cap B_2\neq \o$, we must have 
\begin{equation}
\label{eqref2}
|(A_1\cap B_1)\ast (A_2\cap B_2)|\geq |A_1\cap B_1|\geq |A_2\cap B_2|.
\end{equation}
Therefore, we have:
\begin{align}
\sum_{\substack{(A_2,B_2)\in \mathcal{H}_1\times\mathcal{H}_2\\ A_2\cap B_2\neq \o}}|(A_1\cap B_1)\ast (A_2\cap B_2)|
\geq \sum_{\substack{(A_2,B_2)\in \mathcal{H}_1\times\mathcal{H}_2\\ A_2\cap B_2\neq \o}}|A_1\cap B_1| \geq \sum_{\substack{(A_2,B_2)\in \mathcal{H}_1\times\mathcal{H}_2\\ A_2\cap B_2\neq \o}}|A_2\cap B_2| \label{eqe3}
\end{align}
Now since $\mathcal{H}_1$ and $\mathcal{H}_2$ are two partitions of $Q$, we must have $\displaystyle \sum_{\substack{(A_2,B_2)\in \mathcal{H}_1\times\mathcal{H}_2\\ A_2\cap B_2\neq \o}}|A_2\cap B_2|=|Q|$. We conclude that all the inequalities in \eqref{eqe1}, \eqref{eqe2}, \eqref{eqref2} and \eqref{eqe3} are in fact equalities. Therefore, for all $A_2\in\mathcal{H}_1$ and $B_2\in\mathcal{H}_2$ such that $A_2\cap B_2\neq \o$, we have $|A_2\cap B_2|=|A_1\cap B_1|$ (i.e., $\mathcal{H}_1\wedge\mathcal{H}_2$ is a balanced partition), and $|(A_1\cap B_1)\ast (A_2\cap B_2)|=|(A_1\ast A_2)\cap (B_1\ast B_2)|$. Now \eqref{eqref} implies that $(A_1\cap B_1)\ast (A_2\cap B_2)=(A_1\ast A_2)\cap (B_1\ast B_2)$. Therefore, $(\mathcal{H}_1\wedge\mathcal{H}_2)^{\ast} =\mathcal{H}_1^{\ast}\wedge\mathcal{H}_2^{\ast}$.

If $\mathcal{H}_1$ and $\mathcal{H}_2$ are of periods $n_1$ and $n_2$ respectively, then $\mathcal{H}_1\wedge\mathcal{H}_2$ is a stable partition whose period is at most $\mathrm{lcm}(n_1,n_2)$.
\end{proof}

\section{Polarization process}

In this section, we consider ordinary channels having a quasigroup structure on their input alphabet.

\begin{mydef}
Let $(Q,\ast)$ be an arbitrary quasigroup, and let $P:Q\longrightarrow\mathcal{Y}$ be a single user channel. We define the two channels $P^-:Q\longrightarrow\mathcal{Y}\times\mathcal{Y}$ and $P^+:Q\longrightarrow\mathcal{Y}\times\mathcal{Y}\times Q$ as follows:
$$P^-(y_1,y_2|u_1)=\frac{1}{|Q|}\sum_{u_2\in Q}P(y_1|u_1\ast u_2)P(y_2|u_2),$$
$$P^+(y_1,y_2,u_1|u_2)=\frac{1}{|Q|}P(y_1|u_1\ast u_2)P(y_2|u_2).$$
For any $s=(s_1,\ldots,s_n)\in\{-,+\}^n$, we define $P^s:=((P^{s_1})^{s_2}\ldots)^{s_n}$.
\end{mydef}

\begin{myrem}
\label{rem1}
Let $U_1$ and $U_2$ be two independent random variables uniformly distributed in $Q$. Set $X_1=U_1\ast U_2$ and $X_2=U_2$, then $X_1$ and $X_2$ are independent and uniform in $Q$ since $\ast$ is a quasigroup operation. Let $Y_1$ and $Y_2$ be the outputs of the channel $P$ when $X_1$ and $X_2$ are the inputs respectively. It is easy to see that $I(P^-)=I(U_1;Y_1,Y_2)$ and $I(P^+)=I(U_2;Y_1,Y_2,U_1)$. We have:
\begin{align*}
I(P^-)+I(P^+)&=I(U_1;Y_1,Y_2)+I(U_2;Y_1,Y_2,U_1)=I(U_1,U_2;Y_1,Y_2)\\ &=I(X_1,X_2;Y_1,Y_2)=I(X_1;Y_1)+I(X_2;Y_2)=2I(P).
\end{align*}
It is clear that $$I(P^+)=I(U_2;Y_1,Y_2,U_1)\geq I(U_2;Y_2)=I(X_2;Y_2)=I(P).$$ We conclude that $I(P^-)\leq I(P)\leq I(P^+)$.
\end{myrem}

\begin{mydef}
Let $\mathcal{H}$ be a balanced partition of $(Q,/^{\ast})$, we define the channel $P[\mathcal{H}]:\mathcal{H}\longrightarrow\mathcal{Y}$ by:
$$P[\mathcal{H}](y|H)=\frac{1}{||\mathcal{H}||}\sum_{\substack{x\in Q,\\\emph{Proj}_{\mathcal{H}}(x) = H}}P(y|x).$$
\end{mydef}

\begin{myrem}
If $X$ is a random variable uniformly distributed in $Q$ and $Y$ is the output of the channel $P$ when $X$ is the input, then it is easy to see that $I(P[\mathcal{H}])=I(\emph{Proj}_{\mathcal{H}}(X);Y)$.
\end{myrem}

\begin{mydef}
\label{def1}
Let $\{B_n\}_{n\geq1}$ be i.i.d. uniform random variables in $\{-,+\}$. We define the channel-valued process $\{P_n\}_{n\geq0}$ by:
\begin{align*}
P_0 &:= P,\\
P_{n} &:=P_{n-1}^{B_n}\;\forall n\geq1.
\end{align*}
\end{mydef}

The main result of this section is that almost surely $P_n$ becomes a channel where the output is ``almost equivalent" to the projection of the input onto a stable partition of $(Q,/^{\ast})$:

\begin{mythe}
Let $(Q,\ast)$ be a quasigroup and let $P:Q\longrightarrow\mathcal{Y}$ be an arbitrary channel. Then for any $\delta>0$, we have:
\begin{align*}
\lim_{n\to\infty} \frac{1}{2^n} \bigg|\Big\{ s\in\{-,+\}^n:\; &\exists \mathcal{H}_s\; \emph{a stable partition of $(Q,/^{\ast})$},\\
&\big| I(P^s)-\log|\mathcal{H}_s|\big|<\delta, \big| I(P^s[\mathcal{H}_s])-\log|\mathcal{H}_s|\big|<\delta \Big\}\bigg| = 1.
\end{align*}
\label{mainthe1}
\end{mythe}

\begin{myrem}
Theorem \ref{mainthe1} can be interpreted as follows: in a polarized channel $P^s$, we have $I(P^s)\approx I(P^s[\mathcal{H}_s])\approx \log|\mathcal{H}_s|$ for a certain stable partition $\mathcal{H}_s$ of $(Q,/^{\ast})$. Let $X_s$ and $Y_s$ be the channel input and output of $P^s$ respectively. $I(P^s[\mathcal{H}_s])\approx \log|\mathcal{H}_s|$ means that $Y_s$ ``almost" determines $\mathrm{Proj}_{\mathcal{H}_s}(X_s)$. On the other hand, $I(P^s)\approx I(P^s[\mathcal{H}_s])$ means that there is ``almost" no information about $X_s$ other than $\mathrm{Proj}_{\mathcal{H}_s}(X_s)$ which can be determined from $Y_s$.
\end{myrem}

In order to prove theorem \ref{mainthe1}, we need several lemmas:

\begin{mylem}
\label{lemEqPhi}
Let $(Q,\ast)$ be a quasigroup. If $A$, $B$ and $C$ are three non-empty subsets of $Q$ such that $|A|=|B|=|C|=|A\ast C|=|B\ast C|$, then either $A\cap B=\o$ or $A=B$.
\end{mylem}
\begin{proof}
Suppose that $A\cap B\neq \o$ and let $a\in A\cap B$. The fact that $|A\ast C|=|C|$ implies that $A\ast C= a\ast C$. Similarly, we also have $B\ast C= a\ast C$ since $a\in B$. Therefore, $(A\cup B)\ast C= a\ast C$, and so $|(A\cup B)\ast C|=|C|=|A|$. By noticing that $|A|\leq |A\cup B|\leq |(A\cup B)\ast C|=|A|$, we conclude that $|A\cup B|=|A|$, which implies that $A=B$ since $|A|=|B|$.
\end{proof}

\begin{mydef}
Let $Q$ be a set, and let $A$ be a subset of $Q$, we define the distribution $\mathbb{I}_A$ on $Q$ as $\mathbb{I}_A(x)=\frac{1}{|A|}$ if $x\in A$ and $\mathbb{I}_A(x)=0$ otherwise.
\end{mydef}

\begin{mylem}
\label{lemnorm}
Let $X$ be a random variable on $Q$, and let $A$ be a subset of $Q$. Suppose that there exist $\delta>0$ and an element $a\in A$ such that $|\emph{P}_X(x)-\emph{P}_X(a)|<\delta$ for all $x\in A$ and $\emph{P}_X(x)<\delta$ for all $x\notin A$. Then $||\emph{P}_X-\mathbb{I}_A||_{\infty}<|Q|\delta$.
\end{mylem}
\begin{proof}
We have:
\begin{align*}
\big|1-|A|P_X(a)\big| &= \Big|\Big(\sum_{x\in Q}\textrm{P}_X(x)\Big)-|A|\textrm{P}_X(a)\Big| =\Big|\sum_{x\in A}\big(\textrm{P}_X(x)-\textrm{P}_X(a)\big)+\sum_{x\in A^c}\textrm{P}_X(x)\Big|\\
&\leq \sum_{x\in A}\big|\textrm{P}_X(x)-\textrm{P}_X(a)\big|+\sum_{x\in A^c}\textrm{P}_X(x)<(|Q|-1)\delta.
\end{align*}
Therefore, $\big|\textrm{P}_X(a)-\frac{1}{|A|}\big|<\frac{|Q|-1}{|A|}\delta\leq (|Q|-1)\delta$. Let $x\in A$, then $$\Big|\textrm{P}_X(x)-\frac{1}{|A|}\Big| \leq \big|\textrm{P}_X(x)-\textrm{P}_X(a)\big|+ \Big|\textrm{P}_X(a)-\frac{1}{|A|}\Big|<|Q|\delta.$$
On the other hand, if $x\notin A$ we have $\textrm{P}_X(x)<\delta\leq |Q|\delta$. Thus, $||\emph{P}_X-\mathbb{I}_A||_{\infty}<|Q|\delta$.
\end{proof}

\begin{mydef}
Let $Q$ and $\mathcal{Y}$ be two arbitrary sets. Let $\mathcal{H}$ be a set of subsets of $Q$. Let $(X,Y)$ be a pair of random variables in $Q\times\mathcal{Y}$. We define:
$$\mathcal{A}_{\mathcal{H},\delta}(X,Y)=\Big\{y\in\mathcal{Y}:\;\exists H_y\in\mathcal{H},\; ||\emph{P}_{X|Y=y}-\mathbb{I}_{H_y}||_{\infty}<\delta\Big\},$$
$$\mathcal{P}_{\mathcal{H},\delta}(X;Y)=\emph{P}_Y\big(\mathcal{A}_{\mathcal{H},\delta}(X,Y)\big).$$
\end{mydef}

If $\mathcal{P}_{\mathcal{H},\delta}(X;Y)>1-\delta$ for a small enough $\delta$, then $Y$ is ``almost equivalent" to $\textrm{Proj}_{\mathcal{H}}(X)$. The next lemma shows that if $I(P^-)$ is close to $I(P)$, then the output $Y$ of $P$ is ``almost equivalent" to $\textrm{Proj}_{\mathcal{H}}(X)$, where $X$ is the input to the channel $P$ and $\mathcal{H}$ is a certain balanced partition of $Q$.

\begin{mylem}
\label{lemBalPart}
Let $Q$ and $\mathcal{Y}$ be two arbitrary sets with $|Q|\geq 2$. Let $(X,Y)$ be a pair of random variables in $Q\times\mathcal{Y}$ such that $X$ is uniform. Let $\mathcal{H}$ be a set of disjoint subsets of $Q$ that have the same size. If $\mathcal{P}_{\mathcal{H},\frac{1}{|Q|^2}}(X;Y)>1-\frac{1}{|Q|^2}$, then $\mathcal{H}$ is a balanced partition of $Q$.
\end{mylem}
\begin{proof}
We only need to show that $\mathcal{H}$ covers $Q$. Suppose that there exists $x\in Q$ such that there is no $H$ in $\mathcal{H}$ such that $x\in H$. Then for all $y\in \mathcal{A}_{\mathcal{H},\frac{1}{|Q|^2}}(X,Y)$, $\textrm{P}_{X|Y}(x|y)<\frac{1}{|Q|^2}$. We have:
\begin{align*}
\textrm{P}_X(x)=&\sum_{y\in \mathcal{A}_{\mathcal{H},\frac{1}{|Q|^2}}(X,Y)}  \textrm{P}_{X|Y}(x|y)\textrm{P}_Y(y) + \sum_{y\in \mathcal{A}_{\mathcal{H},\frac{1}{|Q|^2}}(X,Y)^c}  \textrm{P}_{X|Y}(X|Y)\textrm{P}_Y(Y)\\
<&\; \frac{1}{|Q|^2}\textrm{P}_Y\big(\mathcal{A}_{\mathcal{H},\frac{1}{|Q|^2}}(X,Y)\big)+\textrm{P}_Y\big(\mathcal{A}_{\mathcal{H},\frac{1}{|Q|^2}}(X,Y)^c\big) < \frac{1}{|Q|^2} + \frac{1}{|Q|^2}=\frac{2}{|Q|^2}\leq \frac{1}{|Q|}.
\end{align*}
which is a contradiction since $X$ is uniform in $Q$. Therefore, $\mathcal{H}$ covers $Q$ and so it is a balanced partition of $Q$. 
\end{proof}

\begin{mylem}
Let $Q$ and $\mathcal{Y}$ be two arbitrary sets with $|Q|\geq 2$, and let $\mathcal{H}$ and $\mathcal{H}'$ be two balanced partitions of $Q$. Let $(X,Y)$ be a pair of random variables in $Q\times \mathcal{Y}$ such that $X$ is uniform. If $\mathcal{P}_{\mathcal{H},\frac{1}{|Q|^2}}(X;Y)>1-\frac{1}{2|Q|^2}$ and $\mathcal{P}_{\mathcal{H}',\frac{1}{|Q|^2}}(X;Y)>1-\frac{1}{2|Q|^2}$,
then $\mathcal{H}=\mathcal{H}'$.
\label{lemPart}
\end{mylem}
\begin{proof}
Define $\mathcal{H}''=\mathcal{H}\cap \mathcal{H}'$. Let $y\in \mathcal{A}_{\mathcal{H},\frac{1}{|Q|^2}}(X,Y)\cap \mathcal{A}_{\mathcal{H}',\frac{1}{|Q|^2}}(X,Y)$, choose $H\in\mathcal{H}$ and $H'\in\mathcal{H}'$ such that $||\emph{P}_{X|Y=y}-\mathbb{I}_{H}||_{\infty}<\frac{1}{|Q|^2}$ and $||\emph{P}_{X|Y=y}-\mathbb{I}_{H'}||_{\infty}<\frac{1}{|Q|^2}$, then $$||\mathbb{I}_{H'}-\mathbb{I}_{H}||_{\infty}<\frac{2}{|Q|^2}\leq \frac{1}{|Q|}$$ which implies that $H=H'$ and $y\in \mathcal{A}_{\mathcal{H}'',\frac{1}{|Q|^2}}(X,Y)$. Therefore, 
\begin{align*}
\mathcal{P}_{\mathcal{H}'',\frac{1}{|Q|^2}}(X;Y) &\geq \textrm{P}_Y\big( \mathcal{A}_{\mathcal{H},\frac{1}{|Q|^2}}(X,Y)\cap \mathcal{A}_{\mathcal{H}',\frac{1}{|Q|^2}}(X,Y)\big)>1-\frac{1}{|Q|^2}.
\end{align*}
From lemma \ref{lemBalPart} we conclude that $\mathcal{H}''$ is a balanced partition. Therefore, $\mathcal{H}=\mathcal{H}'=\mathcal{H}''$.
\end{proof}

\begin{mylem}
Let $(Q,\ast)$ be a quasigroup with $|Q|\geq 2$, and let $\mathcal{Y}$ be an arbitrary set. For any $\delta>0$, there exists $\epsilon_1(\delta)>0$ depending only on $|Q|$ such that for any two pairs of random variables $(X_1,Y_1)$ and $(X_2,Y_2)$ that are independent and identically distributed in $Q\times \mathcal{Y}$ with $X_1$ and $X_2$ being uniform in $Q$, then $H(X_1\ast X_2|Y_1,Y_2)<H(X_1|Y_1)+\epsilon_1(\delta)$ implies the existence of a balanced partition $\mathcal{H}$ of $Q$ such that $\mathcal{P}_{\mathcal{H},\delta}(X_1;Y_1)>1-\delta$
Moreover, $|H\ast H'|=|H|=|H'|$ for every $H,H'\in\mathcal{H}$.
\label{MainLem}
\end{mylem}
\begin{proof}
Choose $\delta>0$, and let $\delta'=\min\Big\{\frac{\delta}{|Q|^2},\frac{1}{|Q|^4}\Big\}$. Define:
\begin{itemize}
\item $p_{y_1}(x_1):=\textrm{P}_{X_1|Y_1}(x_1|y_1)$ and $p_{y_1,x_2}(x):=p_{y_1}(x/^{\ast}x_2)$.
\item $q_{y_2}(x_2):=\textrm{P}_{X_2|Y_2}(x_2|y_2)$ and $q_{y_2,x_1}(x):=q_{y_2}(x_1\backslash_{\ast} x)$.
\end{itemize}
Note that $q_{y_2}(x_2)=p_{y_2}(x_2)$ since $(X_1,Y_1)$ and $(X_2,Y_2)$ are identically distributed. Nevertheless, we choose to use $q_{y_2}(x_2)$ to denote $\textrm{P}_{X_2|Y_2}(x_2|y_2)$ for the sake of notational consistency.

We have:
\begin{align}
\label{Eqeq1} \textrm{P}_{X_1\ast X_2|Y_1,Y_2}(x|y_1,y_2)&=\sum_{x_1\in Q}p_{y_1}(x_1)q_{y_2,x_1}(x)\\
\label{Eqeq2}&=\sum_{x_2\in Q}q_{y_2}(x_2)p_{y_1,x_2}(x).
\end{align}
Due to the strict concavity of the entropy function, there exists $\epsilon'(\delta')>0$ such that:
\begin{itemize}
\item If $\exists x_1,x_1'\in Q$ such that $p_{y_1}(x_1)\geq \delta'$, $p_{y_1}(x_1')\geq \delta'$ and $||q_{y_2,x_1}-q_{y_2,x_1'}||_{\infty}\geq \delta'$ then 
\begin{equation}
\label{Eqeq3} H(X_1\ast X_2 |Y_1=y_1,Y_2=y_2) \geq H(X_2|Y_2=y_2)+\epsilon'(\delta'),
\end{equation}
(see \eqref{Eqeq1}).
\item If $\exists x_2,x_2'\in Q$ such that $q_{y_2}(x_2)\geq \delta'$, $q_{y_2}(x_2')\geq \delta'$ and $||p_{y_1,x_2}-p_{y_1,x_2'}||_{\infty}\geq \delta'$ then
\begin{equation}
\label{Eqeq4} H(X_1\ast X_2 |Y_1=y_1,Y_2=y_2)\geq H(X_1|Y_1=y_1)+\epsilon'(\delta'),
\end{equation}
(see \eqref{Eqeq2}).
\end{itemize}

Define:
\begin{align*}
&\mathcal{C}_{1}=\bigg\{(y_1,y_2) \in \mathcal{Y}\times \mathcal{Y}:\; \forall x_1,x_1'\in Q,(p_{y_1}(x_1)\geq \delta',\; p_{y_1}(x_1')\geq \delta')\Rightarrow ||q_{y_2,x_1}-q_{y_2,x_1'}||_{\infty}< \delta' \bigg\},
\end{align*}
\begin{align*}
&\mathcal{C}_{2}=\bigg\{(y_1,y_2) \in \mathcal{Y}\times \mathcal{Y}:\; \forall x_2,x_2'\in Q,(q_{y_2}(x_2)\geq \delta',\; q_{y_2}(x_2')\geq \delta')\Rightarrow ||p_{y_1,x_2}-p_{y_1,x_2'}||_{\infty}< \delta' \bigg\}.
\end{align*}
From \eqref{Eqeq3} we have:
\begin{align*}
H( X_1\ast X_2|Y_1,Y_2) &\geq\; H(X_2|Y_2)+\epsilon'(\delta')\textrm{P}_{Y_1,Y_2}(\mathcal{C}_1^c)= H(X_1|Y_1)+\epsilon'(\delta')\textrm{P}_{Y_1,Y_2}(\mathcal{C}_1^c).
\end{align*}
Similarly, from \eqref{Eqeq4} we have $$H(X_1 \ast X_2|Y_1,Y_2)\geq H(X_1|Y_1)+\epsilon'(\delta')\textrm{P}_{Y_1,Y_2}(\mathcal{C}_2^c).$$

Let $\epsilon_1(\delta)=\epsilon'(\delta')\frac{\delta'^2}{2}$, and suppose that $$H(X_1 \ast X_2|Y_1,Y_2)<H(X_1|Y_1)+\epsilon_1(\delta),$$ then we must have $\textrm{P}_{Y_1,Y_2}(\mathcal{C}_1^c)<\frac{\delta'^2}{2}$ and $\textrm{P}_{Y_1,Y_2}(\mathcal{C}_2^c)<\frac{\delta'^2}{2}$, which imply that $\textrm{P}_{Y_1,Y_2}(\mathcal{C})>1-\delta'^2$, where $\mathcal{C}=\mathcal{C}_1\cap\mathcal{C}_2$.

Now for each $a,a',x\in Q$, define:
\begin{itemize}
\item $\pi_{a,a'}(x):=(x\ast a)/^{\ast}a'$, and $\gamma_{a,a'}(x):=a'\backslash_{\ast}(a\ast x)$.
\end{itemize}
And for each $(y_1,y_2)\in \mathcal{Y}\times\mathcal{Y}$, define:
\begin{itemize}
\item $A_{y_1}:=\{x_1\in Q, p_{y_1}(x_1)\geq \delta'\}$.
\item $B_{y_2}:=\{x_2\in Q, q_{y_2}(x_2)\geq \delta'\}$.
\item $\displaystyle a_{y_1}:=\operatorname*{arg\,max}_{x_1}\; p_{y_1}(x_1)$. $\displaystyle b_{y_2}:=\operatorname*{arg\,max}_{x_2}\; q_{y_2}(x_2)$.
\item $H_{y_1,y_2}=\Big\{x_1\in Q:\; \exists b_1,b_1',b_2,b_2',\ldots,b_n,b_n'\in B_{y_2},\; x_1=(\pi_{b_n,b_n'}\circ \ldots \circ \pi_{b_1,b_1'}) (a_{y_1}) \Big\}$.
\item $K_{y_1,y_2}=\Big\{x_2\in Q:\; \exists a_1,a_1',a_2,a_2',\ldots,a_n,a_n'\in A_{y_1},\; x_2= (\gamma_{a_n,a_n'}\circ \ldots \circ \gamma_{a_1,a_1'}) (b_{y_2})\Big\}$.
\end{itemize}
Suppose that $(y_1,y_2)\in \mathcal{C}$. Let $x_1\in H_{y_1,y_2}$, and let $n$ be minimal such that there exists $b_1,b_1',b_2,b_2',\ldots,b_n,b_n'\in B_{y_2}$ satisfying $x_1=(\pi_{b_n,b_n'}\circ \ldots \circ \pi_{b_1,b_1'}) (a_{y_1})$. Define $a_1:=a_{y_1}$, and for $1\leq i\leq n$ define $a_{i+1}=\pi_{b_i,b_i'} (a_i)$, so that $a_{n+1}=x_1$. We must have $a_i\neq a_j$ for $i\neq j$ since $n$ was chosen to be minimal. Therefore, $n+1\leq |Q|$.

For any $1\leq i\leq n$, we have $a_{i+1}=(a_{i}\ast b_i)/^{\ast} b_i'$. Let $x=a_{i}\ast b_i$, then $a_{i+1}=x/^{\ast}b_i'$ and $a_i=x/^{\ast}b_i$. We have $(y_1,y_2)\in \mathcal{C}$, $q_{y_2}(b_i)\geq \delta'$ and $q_{y_2}(b_i')\geq \delta'$, so we must have $||p_{y_1,b_i'}-p_{y_1,b_i}||_{\infty}< \delta'$, and $|p_{y_1,b_i'}(x)-p_{y_1,b_i}(x)|<\delta'$, which implies that $|p_{y_1}(a_{i+1})-p_{y_1}(a_i)|<\delta'$. Therefore:
\begin{equation}
\label{EqEq5}
\begin{split}
|p_{y_1}(x_1)-p_{y_1}(a_{y_1})|&=|p_{y_1}(a_{n+1})-p_{y_1}(a_1)|\leq \sum_{i=1}^n |p_{y_1}(a_{i+1})-p_{y_1}(a_i)|\\
&<n\delta'\leq (|Q|-1)\delta'\leq \frac{|Q|-1}{|Q|^4}<\frac{|Q|-1}{|Q|^2}.
\end{split}
\end{equation}
Since $p_{y_1}(a_{y_1})\geq \frac{1}{|Q|}$, we have $p_{y_1}(x_1)>\frac{1}{|Q|^2} > \delta'$ for every $x_1\in H_{y_1,y_2}$. Therefore, $H_{y_1,y_2}\subset A_{y_1}\;\forall (y_1,y_2)\in \mathcal{C}$. A similar argument yields $K_{y_1,y_2}\subset B_{y_2}\;\forall (y_1,y_2)\in \mathcal{C}$.

Fix two elements $b,b'\in B_{y_2}$. We have $(x_1\ast b)/^{\ast}b'\in H_{y_1,y_2}$ and so $x_1\ast b \in H_{y_1,y_2}\ast b'$ for any $x_1\in H_{y_1,y_2}$. Therefore, $H_{y_1,y_2}\ast b\subset H_{y_1,y_2}\ast b'$. But this is true for any two elements $b,b'\in B_{y_2}$, so $H_{y_1,y_2}\ast b=H_{y_1,y_2}\ast b'$ $\forall b,b'\in B_{y_2}$, and $|H_{y_1,y_2}\ast B_{y_2}|=|H_{y_1,y_2}|$. Similarly, we have $|A_{y_1}\ast K_{y_1,y_2}|=|K_{y_1,y_2}|$. If we also take into consideration the fact that $H_{y_1,y_2}\subset A_{y_1}$ and $K_{y_1,y_2}\subset B_{y_2}$ we conclude:
$$|B_{y_2}|\leq |H_{y_1,y_2}\ast B_{y_2}|=|H_{y_1,y_2}|\leq |A_{y_1}|,$$
$$|A_{y_1}|\leq |A_{y_1}\ast K_{y_1,y_2}|=|K_{y_1,y_2}|\leq |B_{y_2}|.$$
Therefore, $|A_{y_1}|=|H_{y_1,y_2}|=|B_{y_2}|=|K_{y_1,y_2}|$. We conclude that $H_{y_1,y_2}=A_{y_1}$ and $K_{y_1,y_2}=B_{y_2}$. Moreover, we have $|A_{y_1}\ast B_{y_2}|=|A_{y_1}|=|B_{y_2}|$.

Recall that $|p_{y_1}(x_1)-p_{y_1}(a_{y_1})|<(|Q|-1)\delta'$ for all $x_1\in A_{y_1}$ (see \eqref{EqEq5}) and $p_{y_1}(x_1)<\delta'\leq (|Q|-1)\delta'$ for $x_1\notin A_{y_1}$. It is easy to deduce that $$||p_{y_1}-\mathbb{I}_{A_{y_1}}||_{\infty}<|Q|(|Q|-1)\delta'<|Q|^2\delta'.$$ Therefore, $||p_{y_1}-\mathbb{I}_{A_{y_1}}||_{\infty}<\delta$ and $||p_{y_1}-\mathbb{I}_{A_{y_1}}||_{\infty}<\frac{1}{|Q|^2}$. Similarly, $||q_{y_2}-\mathbb{I}_{B_{y_2}}||_{\infty}<\delta$ and $||q_{y_2}-\mathbb{I}_{B_{y_2}}||_{\infty}<\frac{1}{|Q|^2}$.

Now define
$\mathcal{C}_{Y_1}=\Big\{y_1\in\mathcal{Y}:\;\textrm{P}_{Y_2}\big((y_1,Y_2)\in\mathcal{C}\big)>1-\delta'\Big\}$, and for each $y_1\in\mathcal{C}_{Y_1}$, define $$\mathcal{K}_{y_1}=\Big\{y_2\in\mathcal{Y}:\;(y_1,y_2)\in\mathcal{C}\Big\}.$$ Then we have:
$$1-\delta'^2 < \textrm{P}_{Y_1,Y_2}(\mathcal{C})\leq \Big(1-\textrm{P}_{Y_1}(\mathcal{C}_{Y_1})\Big)(1-\delta')+\textrm{P}_{Y_1}(\mathcal{C}_{Y_1}),$$
from which we conclude that $\textrm{P}_{Y_1}(\mathcal{C}_{Y_1})>1-\delta'$. And by definition, we also have $\textrm{P}_{Y_2}(\mathcal{K}_{y_1})>1-\delta'$ for all $y_1\in\mathcal{C}_{Y_1}$. Define $\mathcal{H}_{y_1}=\{B_{y_2}:\;y_2\in\mathcal{K}_{y_1}\}$.

Fix $y_1\in\mathcal{C}_{Y_1}$. Since $|A_{y_1}\ast B|=|A_{y_1}\ast B'|=|A_{y_1}|=|B|=|B'|$ for every $B,B'\in\mathcal{H}_{y_1}$, we conclude that the elements of $\mathcal{H}_{y_1}$ are disjoint and have the same size (lemma \ref{lemEqPhi}). Now since $\textrm{P}_{Y_2}(\mathcal{K}_{y_1})>1-\frac{1}{|Q|^4}$ and since $X_2$ is uniform in $Q$, it is easy to see that $\mathcal{H}_{y_1}$ covers $Q$ and so it is a balanced partition of $Q$ for all $y_1\in\mathcal{C}_{Y_1}$. Moreover, since $\textrm{P}_{Y_2}(\mathcal{K}_{y_1})>1-\frac{1}{|Q|^4}$, we can also conclude that all the balanced partitions $\mathcal{H}_{y_1}$ are the same. Let us denote this common balanced partition by $\mathcal{H}'$.

We have $|A\ast B|=|A|=|B|$ for all $A\in\mathcal{H}$ and all $B\in\mathcal{H}'$, where $\mathcal{H}=\{A_{y_1}:\; y_1\in\mathcal{C}_{Y_1}\}$. By using a similar argument as in the previous paragraph, we can deduce that $\mathcal{H}$ is a balanced partition of $Q$. Moreover, since $(X_1,Y_1)$ and $(X_2,Y_2)$ are identically distributed, we can see that $\mathcal{H}=\mathcal{H}'$. We conclude the existence of a balanced partition $\mathcal{H}$ of $Q$ satisfying $|A\ast B|=|A|=|B|$ for all $A,B\in\mathcal{H}$ and
\begin{align*}
\emph{P}_{Y_1}&\bigg(\Big\{y \in \mathcal{Y}:\; \exists H_y\in\mathcal{H},\; ||\emph{P}_{X_1|Y_1=y}-\mathbb{I}_{H_y}||_{\infty}<\delta\Big\}\bigg)\geq \textrm{P}_{Y_1}(\mathcal{C}_{Y_1})> 1-\delta'> 1-\delta.
\end{align*}
\end{proof}

\begin{mylem}
Let $X_1$ and $X_2$ be two independent random variables in $Q$ such that there exists two sets $A,B\subset Q$ satisfying $||\emph{P}_{X_1}-\mathbb{I}_{A}||_{\infty}<\delta$, $||\emph{P}_{X_2}-\mathbb{I}_{B}||_{\infty}<\delta$ and $|A\ast B|=|A|=|B|$, then $||\emph{P}_{X_1\ast X_2}-\mathbb{I}_{A\ast B}||_{\infty}<2\delta + |Q|\delta^2$.
\label{lemAstProb}
\end{mylem}
\begin{proof}
The fact that $|A\ast B|=|A|=|B|$ implies that for every $x\in A\ast B$, we have $x/^{\ast}b\in A$ for every $b\in B$, and $x/^{\ast}b\in A^c$ for every $b\in B^c$.

For every $a\in Q$ define $\epsilon_{1,a}=\textrm{P}_{X_1}(a)-\frac{1}{|A|}$ if $a\in A$, and $\epsilon_{1,a}=\textrm{P}_{X_1}(a)$ if $a\notin A$. Similarly, for every $b\in Q$ define $\epsilon_{2,b}=\textrm{P}_{X_2}(b)-\frac{1}{|A|}$ if $b\in B$, and $\epsilon_{2,b}=\textrm{P}_{X_1}(b)$ if $b\notin B$.
Let $x\in A\ast B$, we have:
\begin{align*}
\textrm{P}_{X_1\ast X_2}(x)&=\sum_{b\in B} \textrm{P}_{X_1}(x/^{\ast}b)\textrm{P}_{X_2}(b) + \sum_{b\in B^c} \textrm{P}_{X_1}(x/^{\ast}b)\textrm{P}_{X_2}(b)\\
&=\sum_{b\in B} \Big(\frac{1}{|A|}+\epsilon_{1,x/^{\ast}b}\Big)\Big(\frac{1}{|A|}+\epsilon_{2,b}\Big) + \sum_{b\in B^c} \epsilon_{1,x/^{\ast}b}\epsilon_{2,b}\\
&=\frac{1}{|A|} + \frac{1}{|A|}\sum_{b\in B}\Big(\epsilon_{1,x/^{\ast}b}+\epsilon_{2,b}\Big) +\sum_{b\in Q} \epsilon_{1,x/^{\ast}b}\epsilon_{2,b}.
\end{align*}
Therefore,
\begin{align*}
\Big|\textrm{P}_{X_1\ast X_2}(x)-\frac{1}{|A|}\Big|<2\delta + |Q|\delta^2.
\end{align*}
Now let $x\notin A\ast B$, we have:
\begin{align*}
\textrm{P}_{X_1\ast X_2}(x)&=\sum_{b\in B} \textrm{P}_{X_1}(x/^{\ast}b)\textrm{P}_{X_2}(b) + \sum_{\substack{b\notin B\\ x/^{\ast}b\in A}} \textrm{P}_{X_1}(x/^{\ast}b)\textrm{P}_{X_2}(b) + \sum_{\substack{b\notin B\\ x/^{\ast}b\notin A}} \textrm{P}_{X_1}(x/^{\ast}b)\textrm{P}_{X_2}(b)\\
&=\sum_{b\in B} \epsilon_{1,x/^{\ast}b}\Big(\frac{1}{|A|}+\epsilon_{2,b}\Big) + \sum_{\substack{b\notin B\\ x/^{\ast}b\in A}} \Big(\frac{1}{|A|}+\epsilon_{1,x/^{\ast}b}\Big)\epsilon_{2,b} + \sum_{\substack{b\notin B\\ x/^{\ast}b\notin A}} \epsilon_{1,x/^{\ast}b}\epsilon_{2,b} \leq 2\delta + |Q|\delta^2.
\end{align*}
\end{proof}

\begin{mylem}
Let $(Q,\ast)$ be a quasigroup with $|Q|\geq 2$, and let $\mathcal{Y}$ be an arbitrary set. For any $\delta>0$, there exists $\epsilon(\delta)>0$ depending only on $|Q|$ and $\delta$ such that for any channel $P:Q\longrightarrow \mathcal{Y}$, $|I(P^{--})-I(P)|<\epsilon(\delta)$ implies the existence of a balanced partition $\mathcal{H}$ of $Q$ such that $\mathcal{H}^{/^{\ast}}=\{H/^{\ast}H':\; H,H'\in \mathcal{H}\}$ is also a balanced partition of $Q$, $\mathcal{P}_{\mathcal{H},\delta}(X_1;Y_1)>1-\delta$, $\mathcal{P}_{\mathcal{H},\delta}(U_2;Y_1,Y_2,U_1)>1-\delta$ and $\mathcal{P}_{\mathcal{H}^{/^{\ast}},\delta}(U_1;Y_1,Y_2)>1-\delta$. Where $U_1$ and $U_2$ are two independent random variables uniformly distributed in $Q$, $X_1=U_1\ast U_2$, $X_2=U_2$, and $Y_1$ (resp. $Y_2$) is the output of the channel $P$ when $X_1$ (resp. $X_2$) is the input.
\label{MainLem2}
\end{mylem}
\begin{proof}
Let $\delta'=\min\{\delta, \delta'', \frac{1}{16|Q|^2}\}$, where $\delta''>0$ is a small enough number that will be specified later. Let $\epsilon(\delta)=\epsilon_1(\delta')$, where $\epsilon_1$ is given by lemma \ref{MainLem}. Let $P:Q\longrightarrow\mathcal{Y}$ be a channel as in the hypothesis. Then from lemma \ref{MainLem} we conclude the existence of two balanced partitions $\mathcal{H}$ and $\mathcal{H}'$ such that $\mathcal{P}_{\mathcal{H},\delta'}(X_1;Y_1)>1-\delta'$ and $\mathcal{P}_{\mathcal{H}',\delta'}(U_1;Y_1,Y_2)>1-\delta'$.
Moreover, we have $|H_1/^{\ast}H_2|=|H_1|=|H_2|$ for every $H_1,H_2\in\mathcal{H}$.

Given $H\in\mathcal{H}$, define:
\begin{align*}
A_H&= \Big\{y \in \mathcal{Y}:\; ||\textrm{P}_{X_1|Y_1=y}-\mathbb{I}_{H}||_{\infty}<\delta'\Big\} = \Big\{y \in \mathcal{Y}:\; ||\textrm{P}_{X_2|Y_2=y}-\mathbb{I}_{H}||_{\infty}<\delta'\Big\},
\end{align*}
(note that $(X_1,Y_1)$ and $(X_2,Y_2)$ are identically distributed).

Let $x\in H$, we have:
\begin{align*}
\frac{1}{|Q|}&=\;\textrm{P}_{X_1}(x) \\
& = \sum_{y\in\mathcal{A}_{\mathcal{H},\delta'}(X_1;Y_1)\setminus A_H}  \textrm{P}_{X_1|Y_1}(x|y)\textrm{P}_{Y_1}(y) + \sum_{y\in A_h} \textrm{P}_{X_1|Y_1}(x|y)\textrm{P}_{Y_1}(y) + \sum_{y\notin\mathcal{A}_{\mathcal{H},\delta'}(X_1;Y_1)} \textrm{P}_{X_1|Y_1}(x|y)\textrm{P}_{Y_1}(y)\\
& \leq \delta'\mathcal{P}_{\mathcal{H},\delta'}(X_1;Y_1) + \Big(\frac{1}{|H|}+\delta' \Big)\textrm{P}_{Y_1}(A_h) + (1-\mathcal{P}_{\mathcal{H},\delta'}(X_1;Y_1)) < \;2\delta' + 2\textrm{P}_{Y_1}(A_H) \\
&\leq  \frac{1}{8|Q|^2} + 2\textrm{P}_{Y_1}(A_H) < \frac{1}{2|Q|} + 2\textrm{P}_{Y_1}(A_H).
\end{align*}
Therefore, 
\begin{equation}
\label{Eqeq6}
\textrm{P}_{Y_2}(A_H)=\textrm{P}_{Y_1}(A_H) > \frac{1}{4|Q|}.
\end{equation}
Now for each $H_1,H_2\in\mathcal{H}$, define:
\begin{align*}
A'_{H_1,H_2}= \Big\{(y_1,& y_2) \in \mathcal{Y}\times\mathcal{Y}:
||\emph{P}_{U_1|Y_1=y_1,Y_2=y_2}-\mathbb{I}_{H_1/^{\ast}H_2}||_{\infty}<\frac{1}{2|Q|} \Big\}.
\end{align*}
Let $(y_1,y_2)\in A_{H_1}\times A_{H_2}$, then $||\textrm{P}_{X_1|Y_1=y_1}-\mathbb{I}_{H_1}||_{\infty}<\delta'$ and $||\textrm{P}_{X_2|Y_2=y_2}-\mathbb{I}_{H_2}||_{\infty}<\delta'$. Lemma \ref{lemAstProb} implies that
\begin{align*}
||\emph{P}_{U_1|Y_1=y_1,Y_2=y_2}-\mathbb{I}_{H_1/^{\ast}H_2}||_{\infty} &= ||\emph{P}_{X_1/^{\ast}X_2|Y_1=y_1,Y_2=y_2}-\mathbb{I}_{H_1/^{\ast}H_2}||_{\infty}\\
&<2\delta' + |Q|\delta'^2 \leq \frac{1}{8|Q|^2} + |Q|\frac{1}{16^2|Q|^4} < \frac{1}{2|Q|}.
\end{align*}
Therefore, $A_{H_1}\times A_{H_2}\subset A'_{H_1,H_2}$ and so $\textrm{P}_{Y_1,Y_2}(A'_{H_1,H_2}) \geq \textrm{P}_{Y_1}(A_{H_1}) \textrm{P}_{Y_2}(A_{H_2}) > \frac{1}{16|Q|^2}\geq \delta'$ (see \eqref{Eqeq6}). We recall that $\textrm{P}_{Y_1,Y_2}\big(\mathcal{A}_{\mathcal{H}',\delta'}(U_1;Y_1,Y_2)\big)=\mathcal{P}_{\mathcal{H}',\delta'}(U_1;Y_1,Y_2)>1-\delta'$, so $\mathcal{A}_{\mathcal{H}',\delta'}(U_1;Y_1,Y_2) \cap A'_{H_1,H_2}\neq \o$.

Let $(y_1,y_2)\in \mathcal{A}_{\mathcal{H}',\delta'}(U_1;Y_1,Y_2) \cap A'_{H_1,H_2}$, then there exists $H'\in\mathcal{H}'$ such that $||\textrm{P}_{U_1|Y_1=y_1,Y_2=y_2}-\mathbb{I}_{H'}||_{\infty}<\delta'<\frac{1}{2|Q|}$. Now since $(y_1,y_2)\in A'_{H_1,H_2}$, we have $||\textrm{P}_{U_1|Y_1=y_1,Y_2=y_2}-\mathbb{I}_{H_1/^{\ast}H_2}||_{\infty}<\frac{1}{2|Q|}$, so $||\mathbb{I}_{H'}-\mathbb{I}_{H_1/^{\ast}H_2}||_{\infty}<\frac{1}{|Q|}$, we conclude that $H'=H_1/^{\ast}H_2$ and $H_1/^{\ast}H_2\in\mathcal{H}'$. But this is true for any $H_1,H_2\in\mathcal{H}$. Therefore, $\mathcal{H}^{/^{\ast}}\subset\mathcal{H}'$, which implies that $\mathcal{H}^{/^{\ast}}=\mathcal{H}'$ since both $\mathcal{H}'$ and $\mathcal{H}^{/^{\ast}}$ are partitions of $Q$ whose all elements are non-empty. Thus, $$\mathcal{P}_{\mathcal{H},\delta}(X_1;Y_1)\geq \mathcal{P}_{\mathcal{H},\delta'}(X_1;Y_1)>1-\delta'\geq 1-\delta,$$ $$\mathcal{P}_{\mathcal{H}^{/^{\ast}},\delta}(U_1;Y_1,Y_2)\geq \mathcal{P}_{\mathcal{H}^{/^{\ast}},\delta'}(U_1;Y_1,Y_2)>1-\delta'\geq 1-\delta.$$ It remains to prove that $\mathcal{P}_{\mathcal{H},\delta}(U_2;Y_1,Y_2,U_1)>1-\delta$. Define:
\begin{align*}
\mathcal{K}=\mathcal{A}_{\mathcal{H}^{/^{\ast}},\delta''}(U_1;Y_1,Y_2)\cap\Big(\mathcal{A}_{\mathcal{H},\delta''}(X_1;Y_1)\times \mathcal{A}_{\mathcal{H},\delta''}(X_2;Y_2)\Big).
\end{align*}
We have:
\begin{align*}
\textrm{P}_{Y_1}(\mathcal{A}_{\mathcal{H},\delta''}(X_1;Y_1))&=\textrm{P}_{Y_2}(\mathcal{A}_{\mathcal{H},\delta''}(X_2;Y_2))=\mathcal{P}_{\mathcal{H},\delta''}(X_1;Y_1)\geq \mathcal{P}_{\mathcal{H},\delta'}(X_1;Y_1)>1-\delta'\geq 1-\delta''.
\end{align*}
Thus, $\textrm{P}_{Y_1,Y_2}(\mathcal{A}_{\mathcal{H},\delta''}(X_1;Y_1)\times\mathcal{A}_{\mathcal{H},\delta''}(X_2;Y_2))>1-2\delta''$. On the other hand, we have:
\begin{align*}
\textrm{P}_{Y_1,Y_2}(\mathcal{A}_{\mathcal{H}^{/^{\ast}},\delta''}(U_1;Y_1,Y_2))&=\mathcal{P}_{\mathcal{H}^{/^{\ast}},\delta''}(U_1;Y_1,Y_2)=\mathcal{P}_{\mathcal{H}',\delta''}(U_1;Y_1,Y_2)\\
&\geq \mathcal{P}_{\mathcal{H}',\delta'}(U_1;Y_1,Y_2)>1-\delta'\geq 1-\delta'',
\end{align*}
we conclude that $\textrm{P}_{Y_1,Y_2}(\mathcal{K})>1-3\delta''$.
Define:
\begin{align*}
\mathcal{B}=\Big\{(y_1,y_2,u_1)\in\mathcal{Y}\times \mathcal{Y}\times Q:\; (y_1,y_2)\in \mathcal{K},\; &\textrm{and}\; \exists H\in \mathcal{H}^{/^{\ast}},\\
&||\textrm{P}_{U_1|Y_1=y_1,Y_2=y_2}-\mathbb{I}_{H}||_{\infty}<\delta''\;\textrm{and}\; u_1\in H \Big\}.
\end{align*}
If $(y_1,y_2)\in\mathcal{K}$, then $(y_1,y_2)\in \mathcal{A}_{\mathcal{H}^{/^{\ast}},\delta''}(U_1;Y_1,Y_2)$ and so there exists $H_{y_1,y_2}\in\mathcal{H}^{/^{\ast}}$ such that $$||\textrm{P}_{U_1|Y_1=y_1,Y_2=y_2}-\mathbb{I}_{H_{y_1,y_2}}||_{\infty}<\delta'',$$ which implies that $(y_1,y_2,u_1)\in\mathcal{B}$ for all $u_1\in H_{y_1,y_2}$. Now since $||\textrm{P}_{U_1|Y_1=y_1,Y_2=y_2}-\mathbb{I}_{H_{y_1,y_2}}||_{\infty}<\delta''$, it is easy to see that $\textrm{P}_{U_1|Y_1=y_1,Y_2=y_2}(H_{y_1,y_2})\geq 1-|H_{y_1,y_2}|\delta''\geq 1-|Q|\delta''$. Therefore,
\begin{align*}
\textrm{P}_{Y_1,Y_2,U_1}(\mathcal{B})>\textrm{P}_{Y_1,Y_2}(\mathcal{K})(1-|Q|\delta'')>(1-3\delta'')(1-|Q|\delta'')>1-(|Q|+3)\delta''.
\end{align*}
Therefore, if $\delta''\leq \frac{\delta}{|Q|+3}$, then $\textrm{P}_{Y_1,Y_2,U_1}(\mathcal{B})>1-\delta$.

Now let $(y_1,y_2,u_1)\in \mathcal{B}$. There exists $H_1,H_2\in\mathcal{H}$ and $H\in\mathcal{H}^{/^{\ast}}$ such that:
\begin{itemize}
\item $u_1\in H$,
\item $||\textrm{P}_{U_1|Y_1=y_1,Y_2=y_2}-\mathbb{I}_{H}||_{\infty}<\delta''$,
\item $||\textrm{P}_{X_1|Y_1=y_1}-\mathbb{I}_{H_1}||_{\infty}<\delta''$,
\item $||\textrm{P}_{X_2|Y_2=y_2}-\mathbb{I}_{H_2}||_{\infty}<\delta''$.
\end{itemize}
Since $U_1=X_1/^{\ast}X_2$, lemma \ref{lemAstProb} implies that $||\textrm{P}_{U_1|Y_1=y_1,Y_2=y_2}-\mathbb{I}_{H_1/^{\ast}H_2}||_{\infty}<2\delta''+|Q|\delta''^2$, and $||\mathbb{I}_{H}-\mathbb{I}_{H_1/^{\ast}H_2}||_{\infty}<3\delta''+|Q|\delta''^2$. Therefore, if $\delta''\leq \frac{1}{4|Q|}$, then $||\mathbb{I}_{H}-\mathbb{I}_{H_1/^{\ast}H_2}||_{\infty}<\frac{1}{|Q|}$ and $H=H_1/^{\ast}H_2$. Now we have:
\begin{itemize}
\item $u_1\in H$ implies $\big|\textrm{P}_{U_1|Y_1,Y_2}(u_1|y_1,y_2)-\frac{1}{|H|}\big|<\delta''$, i.e., $\frac{1}{|H|}-\delta''<\textrm{P}_{U_1|Y_1,Y_2}(u_1|y_1,y_2)<\frac{1}{|H|}+\delta''$.
\item If $u_2\in H_2$, then $u_1\ast u_2\in H_1$ which implies that $\big|\textrm{P}_{X_1|Y_1}(u_1\ast u_2|y_1)-\frac{1}{|H|}\big|<\delta''$ and $\big|\textrm{P}_{X_2|Y_2}(u_2|y_2)-\frac{1}{|H|}\big|<\delta''$.
\item If $u_2\notin H_2$, then $u_1\ast u_2\notin H_1$, so $\textrm{P}_{X_1|Y_1}(u_1\ast u_2|y_1)<\delta''$ and $\textrm{P}_{X_2|Y_2}(u_2|y_2)<\delta''$.
\end{itemize}
By noticing that
\begin{align*}
\textrm{P}_{U_2|Y_1,Y_2,U_1}&(u_2|y_1,y_2,u_1)=\frac{\textrm{P}_{U_2,U_1|Y_1,Y_2}(u_2,u_1|y_1,y_2)}{\textrm{P}_{U_1|Y_1,Y_2}(u_1|y_1,y_2)}=\frac{\textrm{P}_{X_1|Y_1}(u_1\ast u_2|y_1)\textrm{P}_{X_2|Y_2}(u_2|y_1)}{\textrm{P}_{U_1|Y_1,Y_2}(u_1|y_1,y_2)},
\end{align*}
we  conclude that:
\begin{itemize}
\item If $u_2\in H_2$, we have:
\begin{align*}
\frac{\big(\frac{1}{|H|}-\delta''\big)^2}{\frac{1}{|H|}+\delta''}<&\;
\textrm{P}_{U_2|Y_1,Y_2,U_1}(u_2|y_1,y_2,u_1)<\frac{\big(\frac{1}{|H|}+\delta''\big)^2}{\frac{1}{|H|}-\delta''}.
\end{align*}
\item If $u_2\notin H_2$, we have:
$$\textrm{P}_{U_2|Y_1,Y_2,U_1}(u_2|y_1,y_2,u_1)<\frac{\delta''^2}{\frac{1}{|H|}-\delta''}.$$
\end{itemize}
Consequently, there exists $\beta(\delta)>0$ such that if $\delta''\leq \beta(\delta)$ we get $$||\textrm{P}_{U_2|Y_1=y_1,Y_2=y_2,U_1=u_1}-\mathbb{I}_{H_2}||_{\infty}<\delta.$$ By setting $\delta''=\min\Big\{\frac{\delta}{|Q|+3},\frac{1}{4|Q|},\beta(\delta)\Big\}$, we get $(y_1,y_2,u_1)\in \mathcal{A}_{\mathcal{H},\delta}(U_2;Y_1,Y_2,U_1)$ for every $(y_1,y_2,u_1)\in \mathcal{B}$, i.e., $\mathcal{B}\subset \mathcal{A}_{\mathcal{H},\delta}(U_2;Y_1,Y_2,U_1)$ and $\mathcal{P}_{\mathcal{H},\delta}(U_2;Y_1,Y_2,U_1)\geq \textrm{P}_{Y_1,Y_2,U_1}(\mathcal{B})>1-\delta$.
\end{proof}

Now we are ready to prove theorem \ref{mainthe1}. In fact, we will prove a stronger theorem:

\begin{mythe}
Let $(Q,\ast)$ be a quasigroup and let $P:Q\longrightarrow\mathcal{Y}$ be an arbitrary channel. Then for any $\delta>0$, we have:
\begin{align*}
\lim_{n\to\infty} \frac{1}{2^n} \Bigg|\bigg\{ &s \in\{-,+\}^n:\;\exists \mathcal{H}_s\; \emph{a stable partition of $(Q,/^{\ast})$},\\
&\Big| I(P^s[\mathcal{H}'])-\log\frac{|\mathcal{H}_s|.||\mathcal{H}_s\wedge\mathcal{H}'||}{||\mathcal{H}'||}\Big|<\delta\;\emph{for all stable partitions $\mathcal{H}'$ of $(Q,/^{\ast})$} \bigg\}\Bigg| = 1.
\end{align*}
\label{mainthe11}
\end{mythe}
\begin{proof}
Due to the continuity of the entropy function, there exists $\gamma(\delta)>0$ depending only on $|Q|$ such that if $(X,Y)$ is a pair of random variables in $Q\times\mathcal{Y}$ where $X$ is uniform, and if there exists a stable partition of $\mathcal{H}$ such that $\mathcal{P}_{\mathcal{H},\gamma(\delta)}(X;Y)>1-\gamma(\delta)$, then $\big |I(X;Y)-\log|\mathcal{H}|\big|<\delta$ and $\Big|I\big(\textrm{Proj}_{\mathcal{H}'}(X);Y\big)-\log\frac{|\mathcal{H}|.||\mathcal{H}\wedge\mathcal{H}'||}{||\mathcal{H}'||}\Big|<\delta$ for all stable partitions $\mathcal{H}'$ of $(Q,/^{\ast})$ (remember that $\mathcal{H}\wedge\mathcal{H}'$ is a stable partition by lemma \ref{lemWedge}).

Let $P^n$ be as in definition \ref{def1}. From remark \ref{rem1} we have:
$$\mathbb{E}\big(I(P_{n+1})|P_n\big)=\frac{1}{2}I(P_n^-)+\frac{1}{2}I(P_n^+)=I(P_n)$$
This implies that the process $\{I(P_n)\}_n$ is a martingale, and so it converges almost surely.

Let $m$ be the number of different balanced partitions of $Q$, choose $l>m$ and let $0\leq i\leq l+1$. Almost surely, $|I(P_{n-l+i+1})-I(P_{n-l+i})|$ converges to zero. Therefore, we have:
\begin{align*}
\lim_{n\to\infty} \frac{1}{2^{n-l+i}} |A_{n,l,i}| = 1
\end{align*}
where
\begin{align*}
&A_{n,l,i} :=\Big\{ (s_1,s_2)\in\{-,+\}^{n-l}\times \{-,+\}^{i}:\;|I(P^{(s_1,s_2,-)})-I(P^{(s_1,s_2)})|<\epsilon(\delta')\Big\},
\end{align*}
and $\epsilon(\delta')$ is given by lemma \ref{MainLem2}. Now for each $s_2\in\{-,+\}^i$, define:
\begin{align*}
&A_{n,l,s_2} :=\Big\{ s_1\in\{-,+\}^{n-l}:\; |I(P^{(s_1,s_2,-)})-I(P^{(s_1,s_2)})|<\epsilon(\delta')\Big\}.
\end{align*}
It is easy to see that $\displaystyle |A_{n,l,i}|=\sum_{s_2\in\{-,+\}^i}|A_{n,l,s_2}|$. Therefore,
\begin{align*}
\frac{1}{2^i}\sum_{s_2\in\{-,+\}^i}\Big(\lim_{n\rightarrow\infty}&\frac{1}{2^{n-l}}|A_{n,l,s_2}|\Big)=\lim_{n\to\infty} \frac{1}{2^{n-l+i}} |A_{n,l,i}| = 1,
\end{align*}
i.e., 
\begin{equation}
\label{eqyu}
\sum_{s_2\in\{-,+\}^i}\Big(\lim_{n\rightarrow\infty}\frac{1}{2^{n-l}}|A_{n,l,s_2}|\Big)=2^i.
\end{equation}
On the other hand, it is obvious that $|A_{n,l,s_2}|\leq 2^{n-l}$, and so $\displaystyle \lim_{n\rightarrow\infty}\frac{1}{2^{n-l}}|A_{n,l,s_2}|\leq 1$ for all $s_2\in\{-,+\}^i$. We can now use \eqref{eqyu} to conclude that $\displaystyle \lim_{n\rightarrow\infty}\frac{1}{2^{n-l}}|A_{n,l,s_2}|\leq 1$ for all $s_2\in\{-,+\}^i$.
Therefore, we must have $\displaystyle\lim_{n\to\infty} \frac{1}{2^{n-l}} |A_{n,l}| = 1$,
where
\begin{align*}
A_{n,l} :&= \bigcap_{\substack{0\leq i\leq l+1\\ s_2\in\{-,+\}^i}}A_{n,l,s_2}\\
&=\Big\{ s_1\in\{-,+\}^{n-l}:\; |I(P^{(s_1,s_2,-)})-I(P^{(s_1,s_2)})|<\epsilon(\delta'),\;\forall s_2\in\{-,+\}^i,\; \forall 0\leq i\leq l+1 \Big\}.
\end{align*}
Now define:
\begin{align*}
C_{l} := \Big\{ s_2\in &\{-,+\}^{l}:\textrm{$s_2$ contains the sign $-$ at least $m$ times}\Big\},
\end{align*}
\begin{align*}
B_{n,l} := A_{n,l}\times C_l=\Big\{ s=(s_1,s_2)\in\{-,+\}^{n-l} \times\{-,+\}^{l}:\; s_1\in A_{n,l},\;s_2\in C_{l}\Big\},
\end{align*}
\begin{equation}
\label{Dl}
\begin{split}
D_n := \bigg\{ s \in\{-,+\}^n:\;&\exists \mathcal{H}_s\; \textrm{a stable partition of $(Q,/^{\ast})$},\\
&\Big| I(P^s[\mathcal{H}'])-\log\frac{|\mathcal{H}_s|.||\mathcal{H}_s\wedge\mathcal{H}'||}{||\mathcal{H}'||}\Big|<\delta\;\textrm{for all stable partitions $\mathcal{H}'$ of $(Q,/^{\ast})$} \bigg\}.\\
\end{split}
\end{equation}

Now let $s_1\in A_{n,l}$, let $n-l \leq j\leq n$, let $s=(s_1,s_2)\in \{-,+\}^j$ for some $s_2\in\{-,+\}^{j-n+l}$, let $X_{s}$ be the input to the channel $P^{s}$ and $Y_s$ be the output of it. Since $j-n+l\leq l$, both $s_2$ and $(s_2,-)$ have lengths of at most $l+1$. Therefore, we have $|I(P^{(s_1,s_2,-)})-I(P^{(s_1,s_2)})|<\epsilon(\delta')$ and $|I(P^{(s_1,s_2,-,-)})-I(P^{(s_1,s_2,-)})|<\epsilon(\delta')$. Lemma \ref{MainLem2} implies the existence of a balanced partitions $\mathcal{H}_s$ such that $\mathcal{P}_{\mathcal{H}_s,\delta'}(X_s;Y_s)>1-\delta'$, $\mathcal{P}_{\mathcal{H}_s^{/^{\ast}},\delta'}(X_{(s,-)};Y_{(s,-)})>1-\delta'$ and $\mathcal{P}_{\mathcal{H}_s,\delta'}(X_{(s,+)};Y_{(s,+)})>1-\delta'$ for all $s\in\{-,+\}^j$ ($n-l\leq j\leq n$) having $s_1$ as a prefix. Since $\delta'<\frac{1}{2|Q|^2}$, lemma \ref{lemPart} implies that $\mathcal{H}_{(s,-)}=\mathcal{H}_{s}^{/^{\ast}}$ and $\mathcal{H}_{(s,+)}=\mathcal{H}_{s}$ for all $s\in\{-,+\}^j$ ($n-l\leq j< n$) having $s_1$ as a prefix.

Let $s_2\in C_{l}$, and let $l'$ be the number of $-$ signs in $s_2$ (we have $m\leq l'\leq l$), then there exist $l'+1$ balanced partitions $\mathcal{H}_{i}$  ($0\leq i\leq l'$) such that $\mathcal{H}_{0}=\mathcal{H}_{s_1}$, $\mathcal{H}_{l'}=\mathcal{H}_{(s_1,s_2)}$, and $\mathcal{H}_{i+1}=\mathcal{H}_i^{/^{\ast}}$ for each $0\leq i\leq l'-1$. Since $m$ is the number of different balanced partitions of $Q$, there exist two indices $i$ and $j$ such that $i< j\leq l'$ and $\mathcal{H}_i=\mathcal{H}_j$. We conclude that $\mathcal{H}_{l'}=\mathcal{H}_{(s_1,s_2)}$ is a stable partition of $(Q,/^{\ast})$. Moreover, since $\delta'\leq \gamma(\delta)$, $(s_1,s_2)$ belongs to $D_n$. Therefore, $B_{n,l}\subset D_n$ for any $l\geq m$. Thus:
\begin{align*}
\liminf_{n\to\infty} \frac{1}{2^n} |D_n| &\geq
\lim_{n\to\infty} \frac{1}{2^n} |B_{n,l}|= \lim_{n\to\infty} \Big(\frac{1}{2^{n-l}} |A_{n,l}|\Big)\Big(\frac{1}{2^{l}} |C_l|\Big)=\frac{1}{2^{l}} |C_l|.
\end{align*}
But this is true for any $l\geq m$, we conclude:
\begin{align*}
\liminf_{n\to\infty} \frac{1}{2^n} |D_n| \geq
\lim_{l\to\infty} \frac{1}{2^{l}} |C_l|=1,
\end{align*}
which implies that 
\begin{align*}
\lim_{n\to\infty} \frac{1}{2^n} |D_n|=1.
\end{align*}
\end{proof}

\section{Rate of polarization}

In this section, we are interested in the rate of polarization of $P_n$ to deterministic projection channels.

\begin{mydef}
The \emph{Battacharyya parameter} of an ordinary channel $P$ with input alphabet $\mathcal{X}$ and output alphabet $\mathcal{Y}$ is defined as:
$$Z(P)=\frac{1}{|\mathcal{X}|(|\mathcal{X}|-1)}\sum_{\substack{(x,x')\in\mathcal{X}\times\mathcal{X}\\x\neq x'}}\sum_{y\in\mathcal{Y}}\sqrt{P(y|x)P(y|x')}$$
if $|\mathcal{X}|>1$. And by convention, we take $Z(P)=0$ if $|\mathcal{X}|=1$.
\end{mydef}

It's known that $\textrm{P}_e(P)\leq |\mathcal{X}|Z(P)$ (see \cite{SasogluTelAri}), where $\textrm{P}_e(P)$ is the probability of error of the maximum likelihood decoder of $P$.

\begin{mydef}
Let $(Q,\ast)$ be a quasigroup with $|Q|\geq 2$, and $\mathcal{Y}$ be an arbitrary set. Let $P:Q\longrightarrow \mathcal{Y}$ be an arbitrary channel, and $\mathcal{H}$ be a stable partition of $(Q,/^{\ast})$. We define the channels $P[\mathcal{H}]^-:\mathcal{H}^{/^{\ast}}\longrightarrow\mathcal{Y}\times\mathcal{Y}$ and $P[\mathcal{H}]^+:\mathcal{H}\longrightarrow\mathcal{Y}\times\mathcal{Y}\times\mathcal{H}^{/^{\ast}}$ by:
$$P[\mathcal{H}]^+(y_1,y_2,H_1|H_2)=\frac{1}{|\mathcal{H}|}P[\mathcal{H}](y_1|H_1\ast H_2)P[\mathcal{H}](y_2|H_2),$$
$$P[\mathcal{H}]^-(y_1,y_2|H_1)=\frac{1}{|\mathcal{H}|}\sum_{\substack{H_2\in \mathcal{H}}} P[\mathcal{H}](y_1|H_1\ast H_2)P[\mathcal{H}](y_2|H_2).$$
\end{mydef}

\begin{mylem}
$P[\mathcal{H}]^+$ is degraded with respect to $P^+[\mathcal{H}]$, and $P[\mathcal{H}]^-$ is equivalent to $P^-[\mathcal{H}^{/^{\ast}}]$.
\label{Degraded}
\end{mylem}
\begin{proof}
Let $(H_1,H_2,y_1,y_2)\in \mathcal{H}^{/^{\ast}}\times \mathcal{H}\times\mathcal{Y}\times\mathcal{Y}$, we have:
\begin{align*}
P[\mathcal{H}]^+(y_1,y_2,H_1|H_2)&=\frac{1}{|\mathcal{H}|}P[\mathcal{H}](y_1|H_1\ast H_2)P[\mathcal{H}](y_2|H_2)\\
&= \frac{1}{|Q|.||\mathcal{H}||}\sum_{\substack{x_1\in Q\\ \textrm{Proj}_{\mathcal{H}}(x_1) = H_1\ast H_2}} P(y_1|x_1)\sum_{\substack{x_2\in Q\\ \textrm{Proj}_{\mathcal{H}}(x_2) = H_2}} P(y_2|x_2)\\
&= \frac{1}{|Q|.||\mathcal{H}||}\sum_{\substack{x_1\in Q\\ \textrm{Proj}_{\mathcal{H}^{/^{\ast}}}(x_1) = H_1}} \sum_{\substack{x_2\in Q\\ \textrm{Proj}_{\mathcal{H}}(x_2) = H_2}} P(y_1|x_1\ast x_2)P(y_2|x_2)\\
&= \frac{1}{||\mathcal{H}||} \sum_{\substack{x_1\in Q\\ \textrm{Proj}_{\mathcal{H}^{/^{\ast}}}(x_1) = H_1}} \sum_{\substack{x_2\in Q\\ \textrm{Proj}_{\mathcal{H}}(x_2) = H_2}} P^+(y_1,y_2,x_1|x_2)\\
&= \sum_{\substack{x_1\in Q\\ \textrm{Proj}_{\mathcal{H}^{/^{\ast}}}(x_1) = H_1}} P^+[\mathcal{H}](y_1,y_2,x_1|H_2).
\end{align*}
Therefore, $P[\mathcal{H}]^+$ is degraded with respect to $P^+[\mathcal{H}]$. Now let $(H_1,y_1,y_2)\in \mathcal{H}^{/^{\ast}}\times\mathcal{Y}\times\mathcal{Y}$, we have:
\begin{align*}
P[\mathcal{H}]^-(y_1,y_2|H_1)&=\frac{1}{|\mathcal{H}|}\sum_{\substack{H_2\in \mathcal{H}}} P[\mathcal{H}](y_1|H_1\ast H_2)P[\mathcal{H}](y_2|H_2)\\
&= \frac{1}{|Q|.||\mathcal{H}||}\sum_{\substack{H_2\in \mathcal{H}}}\sum_{\substack{x_1\in Q\\ \textrm{Proj}_{\mathcal{H}}(x_1) = H_1\ast H_2}} P(y_1|x_1)\sum_{\substack{x_2\in Q\\ \textrm{Proj}_{\mathcal{H}}(x_2) = H_2}} P(y_2|x_2)\\
&= \frac{1}{|Q|.||\mathcal{H}||}\sum_{\substack{H_2\in \mathcal{H}}}\sum_{\substack{x_1\in Q\\ \textrm{Proj}_{\mathcal{H}^{/^{\ast}}}(x_1) = H_1}} \sum_{\substack{x_2\in Q\\ \textrm{Proj}_{\mathcal{H}}(x_2) = H_2}} P(y_1|x_1\ast x_2)P(y_2|x_2)\\
&= \frac{1}{|Q|.||\mathcal{H}||}\sum_{\substack{x_1\in Q\\ \textrm{Proj}_{\mathcal{H}^{/^{\ast}}}(x_1) = H_1}} \sum_{x_2\in Q} P(y_1|x_1\ast x_2)P(y_2|x_2)\\
&= \frac{1}{||\mathcal{H}||} \sum_{\substack{x_1\in Q\\ \textrm{Proj}_{\mathcal{H}^{/^{\ast}}}(x_1) = H_1}} P^-(y_1,y_2|x_1) = P^-[\mathcal{H}^{/^{\ast}}](y_1,y_2|H_1).
\end{align*}
Therefore, $P[\mathcal{H}]^-$ is equivalent to $P^-[\mathcal{H}^{/^{\ast}}]$.
\end{proof}

\begin{mydef}
Let $\mathcal{H}$ be a stable partition of $(Q,/^{\ast})$, we define the stable partitions $\mathcal{H}^-$ and $\mathcal{H}^+$, by $\mathcal{H}^{/^{\ast}}$ and $\mathcal{H}$ respectively.
\end{mydef}

\begin{mylem}
Let $B_n$ and $P_n$ be defined as in definition \ref{def1}. For each stable partition $\mathcal{H}$ of $(Q,/^{\ast})$, we define the stable partition-valued process $\mathcal{H}_n$ by:
\begin{align*}
\mathcal{H}_0 &:= \mathcal{H},\\
\mathcal{H}_{n} &:=\mathcal{H}_{n-1}^{B_n}\;\forall n\geq1.
\end{align*}
Then $I(P_n[\mathcal{H}_n])$ converges almost surely to a number in $\mathcal{L}_{\mathcal{H}}:=\big\{\log d:\; d\;\textrm{divides}\;|\mathcal{H}|\big\}$.
\end{mylem}
\begin{proof}
Since $P_{n}[\mathcal{H}_n]^-$ is equivalent to $P_{n}^-[\mathcal{H}_n^{/^{\ast}}]$ and $P_{n}[\mathcal{H}_n]^+$ is degraded with respect to $P_{n}^+[\mathcal{H}_n]$ (lemma \ref{Degraded}), we have:
\begin{align*}
\mathbb{E}\Big( I(P_{n+1} &[\mathcal{H}_{n+1}])\Big|P_n\Big)=\frac{1}{2}I(P_n^{-}[\mathcal{H}_{n}^{/^{\ast}}])+\frac{1}{2}I(P_n^{+}[\mathcal{H}_{n}])\geq \frac{1}{2}I(P_n[\mathcal{H}_{n}]^-)+\frac{1}{2}I(P_n[\mathcal{H}_{n}]^+)=I(P_n[\mathcal{H}_n]).
\end{align*}
This implies that the process $I(P_n[\mathcal{H}_n])$ is a sub-martingale and therefore it converges almost surely. Let $\delta>0$, and define $D_{l,\delta}$ as in \eqref{Dl}, we have shown that $\displaystyle\lim_{n\longrightarrow\infty}\frac{1}{2^n}|D_{n,\delta}|=1$. It is easy to see that almost surely, for every $\delta>0$ and for every $n_0>0$ there exists $n>n_0$ such that $(B_1,\ldots,B_n)\in D_{l,\delta}$.

Let $B_n$ be a realization in which $I(P_n[\mathcal{H}_n])$ converges to a limit $x$, and in which for every $\delta>0$ and for every $n_0>0$ there exists $n>n_0$ such that $(B_1,\ldots,B_n)\in D_{n,\delta}$. Let $\delta>0$ and let $n_0>0$ be chosen such that $|I(P_n[\mathcal{H}_n])-x|<\delta$ for every $n>n_0$. Choose $n>n_0$ such that $(B_1,\ldots,B_n)\in D_{n,\delta}$, this means that there exists a stable partition $\mathcal{H}'$ of $(Q,/^{\ast})$ such that $$\Big| I(P_n[\mathcal{H}_n])-\log\frac{|\mathcal{H}'|.||\mathcal{H}'\wedge\mathcal{H}_n||}{||\mathcal{H}_n||}\Big|<\delta.$$
Therefore,
$\displaystyle\Big| x-\log\frac{|\mathcal{H}_n|.||\mathcal{H}'\wedge\mathcal{H}_n||}{||\mathcal{H}'||}\Big|<2\delta$, which implies that $\displaystyle\Big| x-\log\frac{|\mathcal{H}'|.||\mathcal{H}'\wedge\mathcal{H}_n||}{||\mathcal{H}_n||}\Big|$ since $|Q|=|\mathcal{H}'|.||\mathcal{H}'||=|\mathcal{H}_n|.||\mathcal{H}_n||$.

By noticing that $\frac{|\mathcal{H}_n|.||\mathcal{H}'\wedge\mathcal{H}_n||}{||\mathcal{H}'||}$ divides $|\mathcal{H}_n|=|\mathcal{H}|$, we conclude that $d(x,\mathcal{L}_{\mathcal{H}})<2\delta$ for every $\delta>0$. Therefore, $x\in \mathcal{L}_{\mathcal{H}}$.
\end{proof}

\begin{mylem}
Let $\displaystyle P:Q\rightarrow \mathcal{Y}$ be an ordinary channel where $Q$ is a quasigroup with $|Q|\geq 2$. For any stable partition $\mathcal{H}$ of $(Q,/^{\ast})$, we have:
\begin{align*}
\displaystyle\lim_{n\to\infty}\frac{1}{2^n}\bigg|\Big\{s\in\{-,+\}^n:\; \exists \mathcal{H}\;&\textrm{a stable partition of $(Q,/^{\ast})$},\\
&\;\;\;\;\;\;\;\;\;\;\;\;\;\;\;\;\;\;\;\;\;I(P^s[\mathcal{H}])>\log |\mathcal{H}|-\epsilon,\;Z(P^s[\mathcal{H}])\geq2^{-{2^{n\beta}}}\Big\}\bigg|=0,
\end{align*}
for any $\displaystyle 0<\epsilon< \log 2$ and any $0<\beta<\frac{1}{2}$.
\label{batta}
\end{mylem}
\begin{proof}
Let $\displaystyle 0<\epsilon< \log 2$ and $0<\beta<\frac{1}{2}$, and let $\mathcal{H}$ be a stable partition of $(Q,/^{\ast})$. $I(P_n[\mathcal{H}_n])$ converges almost surely to an element in $\mathcal{L}_{\mathcal{H}}$. Due to the relations between the quantities $I(P)$ and $Z(P)$ (see \emph{proposition 3.3} of \cite{SasogluThesis}) we can see that $Z(P_n[\mathcal{H}_n])$ converges to 0 if and only if $I(P_n[\mathcal{H}_n])$ converges to $\log|\mathcal{H}|$, and there is a number $z_0>0$ such that $\liminf Z(P_n[H])>z_0$ whenever $I(P_n[H])$ converges to a number in $\mathcal{L}_{\mathcal{H}}$ other than $\log|\mathcal{H}|$. Therefore, we can say that almost surely, we have:
$$\lim Z(P_n[\mathcal{H}_n])=0\;\;\textrm{or}\;\;\liminf Z(P_n[H])>z_0$$

$Z(P_n^+[\mathcal{H}_n^+])\leq Z(P_n[\mathcal{H}_n]^+)$ since $P_n[\mathcal{H}_n]^+$ is degraded with respect to $P_n^+[\mathcal{H}_n^+]$, and $Z(P_n^-[\mathcal{H}_n^-])=Z(P_n[\mathcal{H}_n]^-)$ since $P_n[\mathcal{H}_n]^-$ and $P_n^-[\mathcal{H}_n^-]$ are equivalent (see lemma \ref{Degraded}). From lemma 3.5 of \cite{SasogluThesis} we have:
\begin{itemize}
\item $Z(P_n[\mathcal{H}_n]^-)\leq \big(|\mathcal{H}|^2-|\mathcal{H}|+1\big) Z(P_n[\mathcal{H}_n])$.
\item $Z(P_n[\mathcal{H}_n]^+)\leq \big(|\mathcal{H}|-1\big)Z(P_n[\mathcal{H}_n])^2$.
\end{itemize}
Therefore, we have $Z(P_n^-[\mathcal{H}_n])\leq K.Z(P_n[\mathcal{H}_n])$ and $Z(P_n^+[\mathcal{H}_n])\leq K.Z(P_n[\mathcal{H}_n])^2$, where $K$ is equal to $\big(|\mathcal{H}|^2-|\mathcal{H}|+1\big)$. By applying exactly the same techniques that were used to prove theorem 3.5 of \cite{SasogluThesis} we get: $$\displaystyle\lim_{n\to\infty}\textrm{Pr}\Big(I(P_n[\mathcal{H}_n])>\log |\mathcal{H}|-\epsilon, Z(P_n[\mathcal{H}_n])\geq2^{-{2^{n\beta}}}\Big)=0$$ 
But this is true for all stable partitions $\mathcal{H}$. Therefore,
\begin{align*}
\displaystyle\lim_{n\to\infty} \frac{1}{2^n}\bigg|\Big\{s\in\{-,+\}^n:\; \exists \mathcal{H}\; &\textrm{a stable partition of $(Q,/^{\ast})$},\\
&\;\;\;\;\;\;\;\;\;\;\;\;\;\;\;\;\;\;\;\;\;I(P^s[\mathcal{H}^s])>\log |\mathcal{H}|-\epsilon,\; Z(P^s[\mathcal{H}^s])\geq2^{-{2^{n\beta}}}\Big\}\bigg|=0.
\end{align*}
By noticing that for each $s\in\{-,+\}^n$, there exists a stable partition $\mathcal{H}_s$ such that $\mathcal{H}=\mathcal{H}_s^s$, we conclude:
\begin{align*}
\displaystyle\lim_{n\to\infty}\frac{1}{2^n}\bigg|\Big\{s\in\{-,+\}^n:\; \exists \mathcal{H}\; &\textrm{a stable partition of $(Q,/^{\ast})$},\\
&\;\;\;\;\;\;\;\;\;\;\;\;\;\;\;\;\;\;\;\;\;I(P^s[\mathcal{H}])>\log |\mathcal{H}|-\epsilon,\;Z(P^s[\mathcal{H}])\geq2^{-{2^{n\beta}}}\Big\}\bigg|=0.
\end{align*}
\end{proof}

\begin{mythe}
The convergence of $P_n$ to projection channels is almost surely fast:
\begin{align*}
\lim_{n\to\infty} \frac{1}{2^n}\bigg|\big\{ s\in\{-,+\}^n: &\exists \mathcal{H}\;\emph{a stable partition of $(Q,/^{\ast})$},\\
& \big|I(P^s)-\log|\mathcal{H}|\big|<\epsilon, \big|I(P^s[\mathcal{H}])-\log|\mathcal{H}|\big|<\epsilon,\; Z(P^s[\mathcal{H}])<2^{-2^{\beta n}} \big\}\bigg| = 1,
\end{align*}
for any $\displaystyle 0<\epsilon<\log 2$, and any $0<\beta<\frac{1}{2}$.
\label{mainthe2}
\end{mythe}
\begin{proof}
Let $\displaystyle 0<\epsilon<\log 2$, and $0<\beta<\frac{1}{2}$. Define:
\begin{align*}
E_0=\Big\{&s\in\{-,+\}^n:\;\exists\mathcal{H}\; \textrm{a stable partition of } (Q,/^{\ast}),\; I(P^s[\mathcal{H}])>\log |\mathcal{H}|-\epsilon,\;Z(P^s[\mathcal{H}])\geq2^{-{2^{\beta n}}}\Big\},
\end{align*}

\begin{align*}
E_1 =\Big\{&s\in\{-,+\}^n: \exists \mathcal{H}\;\textrm{a stable partition of $(Q,/^{\ast})$},\big|I(P^s)-\log|\mathcal{H}|\big|<\epsilon, \big|I(P^s[\mathcal{H}])-\log|\mathcal{H}|\big|<\epsilon \Big\},
\end{align*}

\begin{align*}
E_2 =\bigg\{ s\in\{-,+\}^n: &\exists \mathcal{H}\;\textrm{a stable partition of $(Q,/^{\ast})$},\\
& \big|I(P^s)-\log|\mathcal{H}|\big|<\epsilon, \big|I(P^s[\mathcal{H}])-\log|\mathcal{H}|\big|<\epsilon,\;Z(P^s[\mathcal{H}])<2^{-2^{\beta n}} \bigg\}.
\end{align*}

It is easy to see that $E_1 \setminus E_0 \subset E_2$ and $|E_2|\geq|E_1|-|E_0|$. By theorem \ref{mainthe1} and lemma \ref{batta} we get:
$$1\geq\lim_{n\to\infty}\frac{1}{2^n}|E_2|\geq\lim_{n\to\infty}\frac{1}{2^n}\big(|E_1|-|E_0|\big)=1-0=1.$$
\end{proof}

\section{Polar code construction}

Choose $\displaystyle 0<\epsilon<\log 2$ and $0<\beta<\beta'<\frac{1}{2}$, let $n$ be an integer such that $$|Q|2^n2^{-2^{\beta' n}}<2^{-2^{\beta n}}\;\;\textrm{and}\;\;\frac{1}{2^n}|E_n|>1-\frac{\epsilon}{2\log |Q|},$$ where

\begin{align*}
E_n=\bigg\{ s\in\{-,+\}^n: &\exists \mathcal{H}\;\textrm{a stable partition of $(Q,/^{\ast})$},\\
& \big|I(P^s)-\log|\mathcal{H}|\big|<\frac{\epsilon}{2}, \big|I(P^s[\mathcal{H}])-\log|\mathcal{H}|\big|<\frac{\epsilon}{2},\;Z(P^s[\mathcal{H}])<2^{-2^{\beta' n}} \bigg\}.
\end{align*}

Such an integer exists due to theorem \ref{mainthe2}. A polar code is constructed as follows: If $s\notin E_n$, let $U_s$ be a frozen symbol, i.e., we suppose that the receiver knows $U_s$. On the other hand, if $s\in E_n$, there exists a stable partition $\mathcal{H}_s$ of $G$, such that $\big|I(P^s)-\log|\mathcal{H}_s|\big|<\frac{\epsilon}{2}$, $\big|I(P^s[\mathcal{H}_s])-\log|\mathcal{H}_s|\big|<\frac{\epsilon}{2}$, and $Z(P^s[\mathcal{H}_s])<2^{-2^{\beta' n}}$. Let $f_s:\mathcal{H}_s\longrightarrow G$ be a frozen mapping (in the sense that the receiver knows $f_s$) such that $f_s(H)\in H$ for all $H\in \mathcal{H}_s$, we call such mapping \emph{a section mapping}. We choose $U_s'$ uniformly in $\mathcal{H}_s$ and we let $U_s=f_s(U_s')$. Note that if the receiver can determine $\textrm{Proj}_{\mathcal{H}_s}(U_s)=U_s'$ accurately, then he can also determine $U_s$ since he knows $f_s$.

Since we are free to choose any value for the frozen symbols and for the section mappings, we will analyse the performance of the polar code averaged on all the possible choices of the frozen symbols and for the section mappings. Therefore, $U_s$ are independent random variables, uniformly distributed in $Q$. If $s\notin E_n$, the receiver knows $U_s$ and there is nothing to decode, and if $s\in E_n$, the receiver has to determine $\textrm{Proj}_{\mathcal{H}_s}(U_s)$ in order to successfully determine $U_s$.

We associate the set $\{-,+\}^n$ with the strict total order $<$ defined as $(s_1,\ldots,s_n)<(s_1',\ldots,s_n')$ if and only if there exists $i\in\{1,\ldots,n\}$ such that $s_i=-$, $s_i'=+$ and $s_j=s_j'\;\forall j>i$.

\subsection{Encoding}
Let $\{P_s\}_{s\in\{-,+\}^n}$ be a set of $2^n$ independent copies of the channel $P$. $P_s$ should not be confused with $P^s$: $P_s$ is a copy of the channel $P$ and $P^s$ is a polarized channel obtained from $P$ as before.

Define $U_{s_1,s_2}$ for $s_1\in\{-,+\}^{l}, s_2\in\{-,+\}^{n-l}$, $0\leq l\leq n$, inductively as:
\begin{itemize}
  \item $U_{\o,s}=U_{s}$ if $l=0$, $s\in\{-,+\}^n$.
  \item $U_{(s_1;-),s_2}=U_{s_1,(s_2;+)}\ast U_{s_1,(s_2;-)}$ if $l>0$, $s_1\in\{-,+\}^{l-1}$, $s_2\in\{-,+\}^{n-l}$.
  \item $U_{(s_1;+),s_2}=U_{s_1,(s_2;+)}$ if $l>0$, $s_1\in\{-,+\}^{l-1}$, $s_2\in\{-,+\}^{n-l}$.
\end{itemize}

We send $U_{s,\o}$ through the channel $P_s$ for all $s\in\{-,+\}^n$. Let $Y_s$ be the output of the channel $P_s$, and let $Y=\{Y_s\}_{s\in\{-,+\}^n}$. We can prove by induction on $l$ that the channel $U_{s_1,s_2}\rightarrow \big(\{Y_s\}_{s\;\textrm{has}\;s_1\;\textrm{as a prefix}},\{U_{s_1,s'}\}_{s'<s_2}\big)$ is equivalent to the channel $P^{s_2}$. In particular, the channel $U_{s}\rightarrow \big(Y,\{U_{s'}\}_{s'<s}\big)$ is equivalent to the channel $P^s$. Figure 1 is an illustration of a polar code construction for $n=2$ (i.e., the block-length is $N=2^2=4$).

\begin{figure}[h]
	\begin{center}
	\includegraphics[scale=0.27]{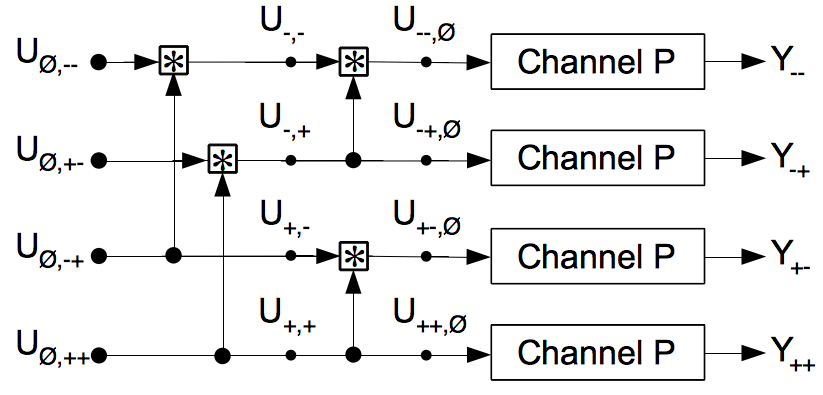}
  \caption{Polar code construction for $n=2$.}  
\label{fig:intro}
\end{center}
\end{figure}

\subsection{Decoding}
If $s\notin E_n$ then the receiver knows $U_s$, there is nothing to decode. Suppose that $s\in E_n$, if we know $\{U_{s'}\}_{s'<s}$ then we can estimate $\textrm{Proj}_{\mathcal{H}_s}(U_s)$ from $\big(Y,\{U_{s'}\}_{s'<s}\big)$ by the maximum likelihood decoder of $P^s[\mathcal{H}_s]$. After that, we estimate $U_s=f_s(\textrm{Proj}_{\mathcal{H}_s}(U_s))$. This motivates the following successive cancellation decoder:

\begin{itemize}
  \item $\hat{U}_s=U_s$ if $s\notin E_n$.
  \item $\hat{U}_s=\mathcal{D}_s(Y,\{\hat{U}_{s'}\}_{s'<s})$ if $s\in E_n$.
\end{itemize}

Where $\mathcal{D}_s(Y,\{U_{s'}\}_{s'<s})$ is the estimate of $U_s$ obtained from $(Y,\{U_{s'}\}_{s'<s})$ by the above procedure.

\subsection{Performance of polar codes}

If $s\in E_n$, the probability of error in estimating $U_s$ is the probability of error in estimating $\textrm{Proj}_{\mathcal{H}_s}(U_s)$ using the maximum likelihood decoder, which is upper bounded by $$|\mathcal{H}_s|.Z(P^s[\mathcal{H}_s])<|Q|2^{-2^{\beta'n}}.$$

Note that $\mathcal{D}_s(Y,\{U_{s'}\}_{s'<s})=U_s\;(\forall s\in E_n) \Leftrightarrow \mathcal{D}_s(Y,\{\hat{U}_{s'}\}_{s'<s})=U_s\;(\forall s\in E_n)$. Therefore, the probability of error of the above successive cancellation decoder is upper bounded by
\begin{align*}
\sum_{s\in E_n} &\textrm{P}\big(\mathcal{D}_s(Y,\{U_{s'}\}_{s'<s})\neq U_{s}\big) < |E_n||Q|2^{-2^{\beta'n}}\leq |Q|2^n2^{-2^{\beta'n}} < 2^{-2^{\beta n}}.
\end{align*}
This upper bound was calculated on average over a random choice of the frozen symbols and of the section mappings. Therefore, there exists at least one choice of the frozen symbols and of the section mappings for which the upper bound of the probability of error still holds.

We should note here that unlike the case of binary input symmetric memoryless channels where the frozen symbols can be chosen arbitrarily, the choice of the frozen symbols and section mappings in our construction of polar codes cannot be arbitrary. The code designer should make sure that his choice of the frozen symbols and section mappings actually yields the desirable probability of error.

The last thing to discuss is the rate of polar codes. The rate at which we are communicating is $\displaystyle R=\frac{1}{2^n}\sum_{s\in E_n}\log |\mathcal{H}_s|$. On the other hand, we have $\big|I(P^s)-\log|\mathcal{H}_s|\big|<\frac{\epsilon}{2}$ for all $s\in E_n$. And since we have $\displaystyle\sum_{s\in\{-,+\}^n} I(P^s)=2^nI(P)$, we conclude:
\begin{align*}
I(P)&=\frac{1}{2^n}\sum_{s\in \{-,+\}^n} I(P^s)= \frac{1}{2^n}\sum_{s\in E_n}I(P^s) + \frac{1}{2^n}\sum_{s\in E_n^c}I(P^s)<\frac{1}{2^n}\sum_{s\in E_n}\Big(\log|\mathcal{H}_s|+\frac{\epsilon}{2}\Big) + \frac{1}{2^n}|E_n^c|\log |Q| \\
&< R + \frac{1}{2^n}|E_n|\frac{\epsilon}{2} + \frac{\epsilon}{2\log |Q|}\log |Q|\leq R+\frac{\epsilon}{2}+\frac{\epsilon}{2}=R+\epsilon.
\end{align*}

To this end we have proven the following theorem which is the main result of this paper:

\begin{mythe}
Let $P:Q\longrightarrow\mathcal{Y}$ be a channel where the input alphabet has a  quasigroup structure. For every $\epsilon>0$ and for every $0<\beta<\frac{1}{2}$, there exists a polar code of length $N$ having a rate $R>I(P)-\epsilon$ and a probability of error $P_e<2^{-N^\beta}$.
\end{mythe}

\section{The case of groups}

\begin{mylem}
Let $(G,\ast)$ be a group, and let $\mathcal{H}$ be a stable partition of $(G,/^{\ast})$. There exists a normal subgroup of $G$ such that $\mathcal{H}$ is the quotient group of $G$ by $H$ (also denoted by $G/H$), and $\emph{Proj}_{\mathcal{H}}(x)=x\bmod H$ for all $x\in G$.
\end{mylem}
\begin{proof}
Let $H$ be the element of $\mathcal{H}$ containing the neutral element $e$ of $G$. For any $H'\in\mathcal{H}$, we have $H'=H'/^{\ast}e\subset H'/^{\ast}H$. Now because of the stability of $\mathcal{H}$, we have $|H'/^{\ast}H|=|H'|$ and so $H'/^{\ast}H=H'$ for all $H'\in\mathcal{H}$. This implies that $\mathcal{H}^{/^{\ast}}=\mathcal{H}$. Now for any $H_1\in\mathcal{H}=\mathcal{H}^{/^{\ast}}$ and $H_2\in\mathcal{H}$, there exists $H_3\in\mathcal{H}$ such that $H_1=H_3/^{\ast}H_2$, and so $H_1\ast H_2=H_3\in\mathcal{H}$. Therefore, we also have $\mathcal{H}^{\ast}=\mathcal{H}$.

Now for any $H'\in\mathcal{H}$, we have $H'=e\ast H'\subset H\ast H'\in\mathcal{H}$, $H'=H'\ast e \subset H'\ast H\in\mathcal{H}$, and $|H'|=|H\ast H'|=|H'\ast H|$, from which we conclude that $H\ast H' = H'\ast H = H'$. This implies that $H\ast H=H$, and $k\ast H = H\ast k$ for any $k\in G$. Therefore, $H$ is a normal subgroup of $G$, and $\mathcal{H}$ is the quotient subgroup of $G$ by $H$.
\end{proof}

By combining the last lemma with theorem \ref{mainthe2}, we get:

\begin{mythe}
Let $P:G\longrightarrow\mathcal{Y}$ be a channel where the input alphabet $G$ has a group structure. $P_n$ converges almost surely to ``homomorphism channels". Moreover, the convergence is almost surely fast:
\begin{align*}
\lim_{n\to\infty} \frac{1}{2^n}\bigg|\big\{ s\in\{-,+&\}^n: \exists H\;\emph{a normal subgroup of $G$},\\
& \big|I(P^s)-\log|G/H|\big|<\epsilon, \big|I(P^s[H])-\log|G/H|\big|<\epsilon,\;Z(P^s[H])<2^{-2^{\beta n}} \big\}\bigg| = 1,
\end{align*}
for any $\displaystyle 0<\epsilon<\log 2$, and any $0<\beta<\frac{1}{2}$. Where $P[H]:G/H\longrightarrow\mathcal{Y}$ is defined as:
$$P[H](y|a)=\frac{1}{|H|}\sum_{\substack{x\in G\\x \bmod H = a}}P(y|x).$$
\label{mainthe22}
\end{mythe}

\section{Polar codes for arbitrary multiple access channels}

In this section, we construct polar codes for an arbitrary multiple access channel, where there is no constraint on the input alphabet sizes: they can be arbitrary, and possibly different from one user to another.

If we have $|\mathcal{X}_k|=p_1^{r_1}p_2^{r_2}\ldots p_{n_k}^{r_{n_k}}$, where $p_1$, \ldots, $p_{n_k}$ are prime numbers, we can assume that $\mathcal{X}_k=\mathbb{F}_{p_1}^{r_1} \mathbb{F}_{p_2}^{r_2}\ldots \mathbb{F}_{p_{n_k}}^{r_{n_k}}$, and so we can replace the $k^{th}$ user by $r_1+r_2+\ldots+r_{n_k}$ virtual users having $\mathbb{F}_{p_1}$, $\mathbb{F}_{p_2}$, \ldots, or $\mathbb{F}_{p_{n_k}}$ as input alphabet respectively. Therefore, we can assume without loss of generality that $\mathcal{X}_k=\mathbb{F}_{q_k}$ for all $k$, where $q_k$ is a prime number. Let $p_1$, $p_2$, \ldots, $p_l$ be the distinct primes which appear in $q_1$, \ldots, $q_m$, and for each $1\leq i\leq l$ let $m_i$ be the number of times $p_i$ appears in $q_1$, \ldots, $q_m$.

We adopt two notations to indicate the users and their inputs:
\begin{itemize}
\item The first notation is the usual one: we have an index $k$ taking value in $\{1,\ldots,m\}$, and the input of the $k^{th}$ user is denoted by $X_k\in\mathbb{F}_{q_k}$.
\item In the second notation, the $m_i$ users having their inputs in $\mathbb{F}_{p_i}$ will be indexed by $(i,1)$, \ldots , $(i,j)$ , \ldots , $(i,m_i)$, where $1\leq i\leq l$ and $1\leq j\leq m_i$. The input of the $(i,j)^{th}$ user is denoted by $X_{i,j}\in\mathbb{F}_{p_i}$. The vector $(X_{i,1},\ldots,X_{i,m_i})\in\mathbb{F}_{p_i}^{m_i}$ is denoted by $\vec{X}_i$.
\end{itemize}

\begin{mydef}
Let $\displaystyle P:\prod_{k=1}^m \mathbb{F}_{q_k}\rightarrow\mathcal{Y}$ be a discrete $m$-user MAC. We define the two channels $\displaystyle P^-:\prod_{k=1}^m \mathbb{F}_{q_k} \rightarrow \mathcal{Y}^2$ and $\displaystyle P^+:\prod_{k=1}^m \mathbb{F}_{q_k} \rightarrow \mathcal{Y}^2\times\prod_{k=1}^m \mathbb{F}_{q_k}$ as:
\begin{align*}
P^-(y_1,y_2|u_1^1,\ldots,u_m^1)=\frac{1}{q_1\ldots q_m}\sum_{(u_1^2,\ldots,u_m^2)\;\in\; \prod_{k=1}^m \mathbb{F}_{q_k}}
P(y_1|u_1^1+u_1^2,\ldots,u_m^1+u_m^2)P(y_2|u_{1}^2,\ldots,u_m^2),
\end{align*}
\begin{align*}
P^+(y_1,y_2,u_1^1,\ldots,u_m^1|u_1^2,\ldots,u_m^2)=\frac{1}{q_1\ldots q_m}P(y_1|u_1^1+u_1^2,\ldots,u_m^1+u_m^2)P(y_2|u_{1}^2,\ldots,u_m^2),
\end{align*}
where the addition $u_k^1+u_k^2$ takes place in $\mathbb{F}_{q_k}$.
\end{mydef}

$P^-$ and $P^+$ can be constructed from two independent copies of $P$ as follows: The $k^{th}$ user chooses independently and uniformly two symbols $U_k^1$ and $U_k^2$ in $\mathbb{F}_{q_k}$, then he calculates $X^1_k=U_k^1+U_k^2$ and $X^2_k=U_k^2$, and he finally sends $X^1_k$ through the first copy of $P$ and $X^2_k$ through the second copy of $P$. Let $Y_1$ and $Y_2$ be the output of the first and second copy of $P$ respectively. $P^-$ is the conditional probability distribution of $Y_1Y_2$ given $U_1^1\ldots U_m^1$, and $P^+$ is the conditional probability distribution of $Y_1Y_2U_1^1\ldots U_m^1$ given $U_1^2\ldots U_m^2$.

Note that the transformation $(U_1^1,\ldots,U_m^1,U_1^2,\ldots,U_m^2)\rightarrow (X_1^1,\ldots,X_m^1,X_1^2,\ldots,X_m^2)$ is bijective and therefore it induces uniform and independent distributions for $X_1^1,\ldots,X_m^1,X_1^2,\ldots,X_m^2$ which are the inputs of the $P$ channels.

\begin{mydef}
Let $\{B_n\}_{n\geq1}$ be i.i.d. uniform random variables on $\{-,+\}$. We define the $MAC$-valued process $\{P_n\}_{n\geq0}$ by:
\begin{align*}
P_0 &:= P,\\
P_{n} &:=P_{n-1}^{B_n}\;\forall n\geq1.
\end{align*}
\end{mydef}

\begin{myprop}
(\cite{SasogluTelYeh} \cite{AbbeTelatar}) The process $\{I[S](P_n)\}_{n\geq0}$ is a bounded super-martingale for all $S\subset\{1,\ldots,m\}$. Moreover, it's a bounded martingale if $S=\{1,\ldots,m\}$.
\end{myprop}
\begin{proof}
\begin{align*}
2I[S](P)&= I[S](P)+I[S](P)=I(X^1(S);Y_1X^1(S^c))+I(X^2(S);Y_2X^2(S^c))\\
&=I(X^1(S)X^2(S);Y_1Y_2X^1(S^c)X^2(S^c))=I(U^1(S)U^2(S);Y_1Y_2U^1(S^c)U^2(S^c))\\
&=I(U^1(S);Y_1Y_2U^1(S^c)U^2(S^c))+I(U^2(S);Y_1Y_2U^1(S^c)U^2(S^c)U^1(S))\\
&\geq I(U^1(S);Y_1Y_2U^1(S^c)) + I(U^2(S);Y_1Y_2U_1^1\ldots U_m^1U^2(S^c))=I[S](P^-) + I[S](P^+).
\end{align*}
Thus, $E\big(I[S](P_{n+1})\big|P_{n}\big)=\frac{1}{2}I[S](P_{n}^-)+\frac{1}{2}I[S](P_{n}^+)\leq I[S](P_n)$, and $\displaystyle I[S](P_n)\leq \sum_{i\in S}\log q_i$ for all $S\subset\{1,\ldots,m\}$, which proves that $\{I[S](P_n)\}_{n\geq0}$ is a bounded super-martingale. If $S=\{1,\ldots,m\}$, the inequality becomes equality, and $\{I[S](P_n)\}_{n\geq0}$ is a bounded martingale.
\end{proof}

From the bounded super-martingale convergence theorem, we deduce that the sequences $\{I[S](P_n)\}_{n\geq0}$ converge almost surely for all $S\subset\{1,\ldots,m\}$.

Since $\frac{1}{2}(I[S](P^-)+I[S](P^+))\leq I[S](P)$ $\forall S\subset \{1,\ldots,m\}$, then $\frac{1}{2}\mathcal{J}(P-)+\frac{1}{2}\mathcal{J}(P+)\subset\mathcal{J}(P)$, but this subset relation can be strict if one of the inequalities is strict for a certain $S\subset\{1,\ldots,m\}$. Nevertheless, for $S=\{1,\ldots,m\}$, we have $\frac{1}{2}(I(P^-)+I(P^+))= I(P)$, so at least one point of the dominant face of $\mathcal{J}(P)$ is present in $\frac{1}{2}\mathcal{J}(P-)+\frac{1}{2}\mathcal{J}(P+)$ since the capacity region is a polymatroid. Therefore, the symmetric sum capacity is preserved, but the dominant face might lose some points.

\begin{mydef}
In order to simplify our notation, we will introduce the notion of generalized matrices:
\begin{itemize}
\item A generalized matrix $\displaystyle A=(A_1,\ldots,A_l) \in \prod_{i=1}^l \mathbb{F}_{p_i}^{m_i\times l_i}$ is a collection of $l$ matrices. $\mathbb{F}_{p_i}^{m_i\times l_i}$ denotes the set of $m_i\times l_i$ matrices with coefficients in $\mathbb{F}_{p_i}$.
\item If $l_i=0$ in $\displaystyle A=(A_1,\ldots,A_l) \in \prod_{i=1}^l \mathbb{F}_{p_i}^{m_i\times l_i}$, we write $A_i=\o$. In case $A_i=\o$ for all $i$, we write $A=\o$.
\item A generalized vector $\displaystyle \vec{x}=(\vec{x}_1,\ldots,\vec{x}_l) \in \prod_{i=1}^l \mathbb{F}_{p_i}^{m_i}$ is a collection of $l$ vectors.
\item Addition of generalized vectors is defined as component-wise addition.
\item The transposition of a generalized matrix is obtained by transposing each matrix of it: $A^T=(A_1^T,\ldots,A_l^T)$.
\item A generalized matrix operates on a generalized vector component-wise: if $\displaystyle A\in \prod_{i=1}^l \mathbb{F}_{p_i}^{m_i\times l_i}$ and $\displaystyle \vec{x}\in \prod_{i=1}^l \mathbb{F}_{p_i}^{m_i}$, then $\displaystyle \vec{y}=A^T\vec{x}\in \prod_{i=1}^l \mathbb{F}_{p_i}^{l_i}$ is defined by $\vec{y}=(A_1^T\vec{x}_1,\ldots,A_l^T\vec{x}_l)$. By convention, we have $\o^T\vec{x}_i=\vec{0}$.
\item A generalized matrix $A$ is said to be full rank if and only if each matrix component of it is full rank.
\item The rank of a generalized matrix $\displaystyle A \in \prod_{i=1}^l \mathbb{F}_{p_i}^{m_i\times l_i}$ is defined by: $\displaystyle\emph{rank}(A)=\sum_{i=1}^l \emph{rank}(A_i)$.
\item The logarithmic rank of a generalized matrix is defined by:
$\displaystyle\emph{lrank}(A)=\sum_{i=1}^l \emph{rank}(A_i).\log p_i$.
\item If $A$ is a generalized matrix satisfying $A_i\neq\o$ and $A_j=\o$ for all $j\neq i$, we say that $A$ is an ordinary matrix and we identify $A$ and $A_i$.
\end{itemize}
\end{mydef}

\begin{mydef}
Let $\displaystyle P:\prod_{i=1}^l \mathbb{F}_{p_i}^{m_i}\rightarrow \mathcal{Y}$ be an $m$-user MAC, let $\displaystyle A\in \prod_{i=1}^l \mathbb{F}_{p_i}^{m_i\times l_i}$ be a full rank generalized matrix. We define the $\emph{rank}(A)$-user MAC $\displaystyle P[A]:\prod_{i=1}^l \mathbb{F}_{p_i}^{l_i}\rightarrow \mathcal{Y}$ as follows:
$$P[A](y|\vec{u})=\frac{1}{\prod_{i=1}^l p_i^{m_i-l_i}}\sum_{\substack{\vec{x}\;\in\;\prod_{i=1}^l \mathbb{F}_{p_i}^{m_i}\\ A^T\vec{x}=\vec{u}}}P(y|\vec{x}).$$
\end{mydef}

The main result of this section is that, almost surely, $P_n$ becomes a channel where the output is ``almost determined by a generalized matrix", and the convergence is almost surely fast:

\begin{mythe}
Let $\displaystyle P:\prod_{i=1}^l \mathbb{F}_{p_i}^{m_i}\rightarrow \mathcal{Y}$ be an $m$-user MAC. Then for every $\displaystyle 0<\epsilon<\log 2$, and for every $0<\beta<\frac{1}{2}$ we have:
\begin{align*}
\lim_{n\to\infty} \frac{1}{2^n}\bigg|\Big\{ s\in\{-,+&\}^n: \exists A_{s}\in\prod_{i=1}^l \mathbb{F}_{p_i}^{m_i\times r_{i,s}},\; A_{s} \emph{ is full rank}, \\
&|I(P^s)-\emph{lrank}(A_s)|<\epsilon,\;|I(P^s[A_s])-\emph{lrank}(A_s)|<\epsilon, Z(P^s[A_s])<2^{-2^{\beta n}} \Big\}\bigg| = 1.
\end{align*}
\label{mainthemac}
\end{mythe}
\begin{proof}
Since $\displaystyle G:=\prod_{i=1}^l \mathbb{F}_{p_i}^{m_i}$ is an abelian group, we can view $P$ as a channel from the Abelian group $G$ to $\mathcal{Y}$. Note that any subgroup of an Abelian group is normal. Therefore, from theorem \ref{mainthe22} we have:
\begin{align*}
\lim_{n\to\infty} \frac{1}{2^n}\bigg|\Big\{  s\in\{-&,+\}^n: \exists H_s\;\textrm{subgroup of $G$},\\
& \big|I(P^s)-\log|G/H_s|\big|<\epsilon, \big|I(P^s[H_s])-\log|G/H_s|\big|<\epsilon,\; Z(P^s[H_s])<2^{-2^{\beta n}} \Big\}\bigg| = 1.
\end{align*}
Let $s\in\{-,+\}^n$ such that that there exists a subgroup $H_s$ of $G$ satisfying:
\begin{itemize}
\item $\big|I(P^s)-\log|G/H_s|\big|<\epsilon$.
\item $\big|I(P^s[H_s])-\log|G/H_s|\big|<\epsilon$.
\item $Z(P^s[H_s])<2^{-2^{\beta n}}$.
\end{itemize}

From the properties of abelian groups, there exist $l$ integers: $r_{1,s}\leq m_1$, \ldots, and $r_{l,s}\leq m_l$ such that $G/H_s$ is isomorphic to $\displaystyle \prod_{i=1}^l \mathbb{F}_{p_i}^{r_{i,s}}$ (Note that $r_{i,s}$ can be zero). Therefore, there exists a surjective homomorphism $\displaystyle f_s: \prod_{i=1}^l \mathbb{F}_{p_i}^{m_i}\longrightarrow \prod_{i=1}^l \mathbb{F}_{p_i}^{r_{i,s}}$, such that for any $\displaystyle\vec{x}\in  \prod_{i=1}^l \mathbb{F}_{p_i}^{m_i}$, $f_s(\vec{x})$ can be determined from $\vec{x}\bmod H_s$ and vice versa.

For all $1\leq i\leq l$, and all $1\leq j\leq m_i$, define the vector $\displaystyle \vec{e}^{\;i,j}\in \prod_{i=1}^l \mathbb{F}_{p_i}^{m_i}$ as having all its components as zeros except the $(i,j)^{th}$ component which is equal to $1$. The order of $\vec{e}^{\;i,j}$ is $p_i$. Let $\displaystyle\vec{y}^{\;i,j}=(\vec{y}^{\;i,j}_{1},\vec{y}^{\;i,j}_{2},\ldots,\vec{y}^{\;i,j}_{l})=f_s(\vec{e}^{\;i,j})\in \prod_{i=1}^l \mathbb{F}_{p_i}^{r_{i,s}}$, if $\vec{y}^{\;i,j}\neq \vec{0}$ then the order of $\vec{y}^{\;i,j}$ must be equal to $p_i$. If $\vec{y}^{\;i,j}_{i'}\neq \vec{0}$ for a certain $i'\neq i$, then $p_{i'}$ divides the order of $\vec{y}^{\;i,j}$ which is a contradiction. Therefore, we must have $\vec{y}^{\;i,j}_{i'}= \vec{0}$ for all $i'\neq i$.

Now for any $\displaystyle\vec{x}\in \prod_{i=1}^l \mathbb{F}_{p_i}^{m_i}$, we have $\displaystyle \vec{x}=\sum_{i=1}^l\sum_{j=1}^{m_i}x_{i,j}\vec{e}^{\;i,j}$, therefore, $\displaystyle f_s(\vec{x})=\sum_{i=1}^l\sum_{j=1}^{m_i}x_{i,j}\vec{y}^{\;i,j}$. Since $\vec{y}^{\;i,j}_{i'}= 0$ for all $i'\neq i$, then $f_s(\vec{x})=A_s^T\vec{x}$, where $\displaystyle A_s=(A_{1,s},\ldots,A_{l,s})\in \prod_{i=1}^l \mathbb{F}_{p_i}^{m_i\times r_{i,s}}$ is a generalized matrix whose components are given by $A_{i,s}=[\vec{y}^{\;i,1}_{i}\; \vec{y}^{\;i,2}_{i}\;\ldots\; \vec{y}^{\;i,m_i}_{i}]^T$. $A_s$ is full rank since $f_s$ is surjective. Moreover, we have:
$$\textrm{lrank}(A_s)=\sum_{i=1}^l r_{i,s}.\log p_i=\log\Big(\prod_{i=1}^l p_i^{r_{i,s}}\Big)=\log|G/H_s|.$$
Recall that for any $\displaystyle\vec{x}\in  \prod_{i=1}^l \mathbb{F}_{p_i}^{m_i}$, $A_s^T\vec{x}=f_s(\vec{x})$ can be determined from $\vec{x}\bmod H_s$ and vice versa, we conclude that $P^s[H_s]$ is equivalent to $P^s[A_s]$. Therefore:
\begin{align*}
\lim_{n\to\infty} \frac{1}{2^n}\bigg|\Big\{ s\in\{-,+&\}^n: \exists A_{s}\in\prod_{i=1}^l \mathbb{F}_{p_i}^{m_i\times r_{i,s}},\; A_{s} \textrm{ is full rank}, \\
&|I(P^s)-\textrm{lrank}(A_s)|<\epsilon,\;|I(P^s[A_s])-\textrm{lrank}(A_s)|<\epsilon, Z(P^s[A_s])<2^{-2^{\beta n}} \Big\}\bigg| = 1.
\end{align*}
\end{proof}

\subsection{Polar code construction for MACs}

Choose $\displaystyle 0<\epsilon<\log 2$, $0<\beta<\beta'<\frac{1}{2}$, and let $n$ be an integer such that
\begin{itemize}
\item  $\displaystyle\Big(\prod_{i=1}^l p_i^{m_i}\Big)2^n2^{-2^{\beta' n}}<2^{-2^{\beta n}}$.
\item $\displaystyle \frac{1}{2^n}|E_n|>1-\frac{\epsilon}{\displaystyle 2\sum_{i=1}^l m_i\log p_i}$.
\end{itemize}
 where
\begin{align*}
E_n =\bigg\{ s\in\{-,+\}^n: &\exists A_{s}\in\prod_{i=1}^l \mathbb{F}_{p_i}^{m_i\times r_{i,s}}, A_{s} \textrm{ is full rank},\\
& |I(P^s)-\textrm{lrank}(A_s)|<\frac{\epsilon}{2},\;|I(P^s[A_s])-\textrm{lrank}(A_s)|<\frac{\epsilon}{2},\;Z(P^s[A_s])<2^{-2^{\beta' n}} \bigg\}.
\end{align*}

Such an integer exists due to theorem \ref{mainthemac}.

For each $s\in\{-,+\}^n$, if $s\notin E_n$ set $F(s,i,j)=1\;\forall i\in\{1,\ldots,l\}\;\forall j\in\{1,\ldots,m_i\}$, and if $s\in E_n$ choose a generalized matrix $A_s=(A_{1,s},\ldots,A_{l,s})$ that satisfies the conditions in $E_n$. For each $1\leq i\leq l$ choose a set of $r_{i,s}$ indices $$S_{i,s}=\{j_1,\ldots j_{r_{i,s}}\}\subset\{1,\ldots,m_i\}$$ such that the corresponding rows of $A_{i,s}$ are linearly independent, then set $F(s,i,j)=1$ if $j\notin S_{i,s}$, and $F(s,i,j)=0$ if $j\in S_{i,s}$. $F(s,i,j)=1$ indicates that the user $(i,j)$ is frozen in the channel $P^s$, i.e., no useful information is being sent.

A polar code is constructed as follows: The user $(i,j)$ sends a symbol $U_{s,i,j}$ through a channel equivalent to $P^s$. If $F(s,i,j)=0$, $U_{s,i,j}$ is an information symbol, and if $F(s,i,j)=1$, $U_{s,i,j}$ is a certain frozen symbol. Since we are free to choose any value for the frozen symbols, we will analyse the performance of the polar code averaged on all the possible choices of the frozen symbols, so we will consider that $U_{s,i,j}$ are independent random variables, uniformly distributed in $\mathbb{F}_{p_i}$ $\forall s\in\{-,+\}^n,\forall i\in\{1,\ldots,l\},\forall j\in\{1,\ldots,m_i\}$. However, the value of $U_{s,i,j}$ will be revealed to the receiver if $F(s,i,j)=1$, and if $F(s,i,j)=0$ the receiver has to estimate $U_{s,i,j}$ from the output of the channel.

We associate the set $\{-,+\}^n$ with the same strict total order $<$ that we defined earlier. Namely, $s_1\ldots s_n<s_1'\ldots s_n'$ if and only if there exists $i\in\{1,\ldots,n\}$ such that $s_i=-$, $s_i'=+$ and $s_j=s_j'\;\forall j>i$.

\subsubsection{Encoding}
Let $\{P_s\}_{s\in\{-,+\}^n}$ be a set of $2^n$ independent copies of the channel $P$. $P_s$ should not be confused with $P^s$: $P_s$ is a copy of the channel $P$ and $P^s$ is a polarized channel obtained from $P$ as before.

Define $U_{s_1,s_2,i,j}$ for $s_1\in\{-,+\}^{l'},s_2\in\{-,+\}^{n-l'}$, $0\leq l'\leq n$ inductively as:
\begin{itemize}
  \item $U_{\o,s,i,j}=U_{s,i,j}$ if $l'=0$, $s\in\{-,+\}^n$.
  \item $U_{(s_1;-),s_2,i,j}=U_{s_1,(s_2;+),i,j}+U_{s_1,(s_2;-),i,j}$ if $l'>0$, $s_1\in\{-,+\}^{l'-1}$, $s_2\in\{-,+\}^{n-l'}$.
  \item $U_{(s_1;+),s_2,i,j}=U_{s_1,(s_2;+),i,j}$ if $l'>0$, $s_1\in\{-,+\}^{l'-1}$, $s_2\in\{-,+\}^{n-l'}$.
\end{itemize}

The user $(i,j)$ sends $U_{s,\o,i,j}$ through the channel $P_s$ for all $s\in\{-,+\}^n$. Let $Y_s$ be the output of the channel $P_s$, and let $Y=\{Y_s\}_{s\in\{-,+\}^n}$. We can prove by induction on $l'$ that the channel $\vec{U}_{s_1,s_2}\rightarrow \big(\{Y_s\}_{s\;\textrm{has}\;s_1\;\textrm{as a prefix}},\{\vec{U}_{s'}\}_{s'<s_2}\big)$ is equivalent to $P^{s_2}$. In particular, the channel $\vec{U}_{s}\rightarrow \big(Y,\{\vec{U}_{s'}\}_{s'<s}\big)$ is equivalent to the channel $P^s$.

\subsubsection{Decoding}
If $s\notin E_n$ then $F(s,i,j)=1$ for all $(i,j)$, and the receiver knows all $U_{s,i,j}$, there is nothing to decode. Suppose that $s\in E_n$, if we know $\{\vec{U}_{s'}\}_{s'<s}$ then we can estimate $\vec{U}_s$ as follows:

\begin{itemize}
  \item If $F(s,i,j)=1$ then we know $U_{s,i,j}$.
  \item We have $F(s,i,j)=0$ for $r_{i,s}$ values of $j$ corresponding to $r_{i,s}$ linearly independent rows of $A_{i,s}$. So if we know $A_{i,s}^T\vec{U}_s$, we can recover $U_{s,i,j}$ for the indices $j$ satisfying $F(s,i,j)=0$.
  \item Since $A_s^T\vec{U}_s\longrightarrow \big(Y,\{\vec{U}_{s'}\}_{s'<s}\big)$ is equivalent to $P^s[A_s]$, we can estimate $A_s^T\vec{U}_s$ using the maximum likelihood decoder of the channel $P^s[A_s]$.
  \item Let $\mathcal{D}_s(Y,\{\vec{U}_{s'}\}_{s'<s})$ be the estimate of $\vec{U}_s$ obtained from $(Y,\{\vec{U}_{s'}\}_{s'<s})$ by the above procedure.
\end{itemize}

This motivates the following successive cancellation decoder:

\begin{itemize}
  \item $\hat{\vec{U}}_s=\vec{U}_s$ if $s\notin E_n$.
  \item $\hat{\vec{U}}_s=\mathcal{D}_s(Y,\{\hat{\vec{U}}_{s'}\}_{s'<s})$ if $s\in E_n$.
\end{itemize}

\subsubsection{Performance of polar codes}

If $s\in E_n$, the probability of error in estimating $A_s^T\vec{U}_s$ using the maximum likelihood decoder is upper bounded by $\displaystyle\Big(\prod_{i=1}^l p_i^{r_{i,s}}\Big)Z(P^s[A_s])< \Big(\prod_{i=1}^l p_i^{m_i}\Big)2^{-2^{\beta' n}}$. Therefore, the probability of error in estimating $\vec{U}_s$ from $(Y,\{\vec{U}_{s'}\}_{s'<s})$ is upper bounded by $\displaystyle \Big(\prod_{i=1}^l p_i^{m_i}\Big)2^{-2^{\beta' n}}$ when $s\in E_n$.

Note that $\mathcal{D}_s(Y,\{\vec{U}_{s'}\}_{s'<s})=\vec{U}_s,\;(\forall s\in E_n) \Leftrightarrow \mathcal{D}_s(Y,\{\hat{\vec{U}}_{s'}\}_{s'<s})=\vec{U}_s\;(\forall s\in E_n)$, so the probability of error of the above successive cancellation decoder is upper bounded by
\begin{align*}
& \sum_{s\in E_n} \textrm{P}\big(\mathcal{D}_s(Y,\{\vec{U}_{s'}\}_{s'<s})\neq \vec{U}_{s}\big) < |E_n|\Big(\prod_{i=1}^l p_i^{m_i}\Big)2^{-2^{\beta' n}}\leq \Big(\prod_{i=1}^l p_i^{m_i}\Big)2^n2^{-2^{\beta' n}} < 2^{-2^{\beta n}}.
\end{align*}
The above upper bound was calculated on average over a random choice of the frozen symbols. Therefore, there is at least one choice of the frozen symbols for which the upper bound of the probability of error still holds.

The last thing to discuss is the rate vector of polar codes. The rate at which the user $(i,j)$ is communicating is $\displaystyle R_{i,j}=\frac{1}{2^n}\sum_{s\in E_n}\big(1-F(s,i,j)\big)\log p_i$, the sum rate is:
\begin{align*}
R&=\sum_{1\leq i\leq l}\sum_{1\leq j\leq m_i}R_{i,j}=\frac{1}{2^n}\sum_{1\leq i\leq l}\sum_{1\leq j\leq m_i}\sum_{s\in E_n}\big(1-F(s,i,j)\big)\log p_i\\
&=\frac{1}{2^n}\sum_{s\in E_n}\sum_{1\leq i\leq l}r_{i,s}\log p_i = \frac{1}{2^n}\sum_{s\in E_n}\textrm{lrank}(A_s).
\end{align*}
We have $|I(P^s)-\textrm{lrank}(A_s)|<\frac{\epsilon}{2}$ and $I(P^s)<\textrm{lrank}(A_s)+\frac{\epsilon}{2}$ for all $s\in E_n$. And since we have $\displaystyle\sum_{s\in\{-,+\}^n} I(P^s)=2^nI(P)$ we conclude:
\begin{align*}
I(P)&=\frac{1}{2^n}\sum_{s\in \{-,+\}^n} I(P^s) = \frac{1}{2^n}\sum_{s\in E_n}I(P^s) + \frac{1}{2^n}\sum_{s\in E_n^c}I(P^s)\\
&<\frac{1}{2^n}\sum_{s\in E_n}\Big(\textrm{lrank}(A_s)+\frac{\epsilon}{2}\Big) + \frac{1}{2^n}|E_n^c|\sum_{i=1}^l m_i\log p_i \\&< R + \frac{1}{2^n}|E_n|\frac{\epsilon}{2} + \frac{\epsilon}{\displaystyle 2\sum_{i=1}^l m_i\log p_i}\sum_{i=1}^l m_i\log p_i\leq R+\frac{\epsilon}{2}+\frac{\epsilon}{2}=R+\epsilon.
\end{align*}

To this end we have proven the following theorem which is the main result of this subsection:

\begin{mythe}
Let $\displaystyle P:\prod_{i=1}^l \mathbb{F}_{p_i}^{m_i}\rightarrow \mathcal{Y}$ be an $m$-user MAC. For every $\displaystyle \epsilon>0$ and for every $0<\beta<\frac{1}{2}$, there exists a polar code of length $N$ having a sum rate $R>I(P)-\epsilon$ and a probability of error $P_e<2^{-N^\beta}$.
\end{mythe}

Note that by changing our choice of the indices in $S_{i,s}$, we can achieve all the portion of the dominant face of the capacity region that is achievable by polar codes. However, this portion of the dominant face that is achievable by polar codes can be strictly smaller than the dominant face. In such case, we say that we have a loss in the dominant face.

\section{Case study}

In this section, we are interested in studying the problem of loss in the capacity region by polarization in a special case of MACs, namely, the MACs that are combination of linear channels which are defined below. For simplicity, we will consider MACs where the input alphabet size is a prime number $q$ and which is the same for all the users. Moreover we will use the base-$q$ logarithm in the expression of the mutual information and entropies.

\begin{mydef}
An $m$-user MAC $P$ is said to be a combination of $l$ linear channels, if there are $l$ matrices $A_1,\ldots,A_l$, $(A_k\in\mathbb{F}_q^{m\times m_k})$ such that $P$ is equivalent to the channel $\displaystyle P_{lin}:\mathbb{F}_q^m\rightarrow\bigcup_{k=1}^l \Big(\{k\}\times \mathbb{F}_q^{m_k}\Big)$ defined by:

\begin{equation*}
P_{lin}(k,\vec{y}|\vec{x})=
\begin{cases}
p_k\;&\emph{\textrm{if}}\;A_k^T\vec{x}=\vec{y}\\
0\;&\emph{\textrm{otherwise}}
\end{cases}
\end{equation*}

where $\displaystyle \sum_{k=1}^l p_k=1$ and $p_k\neq0\;\forall k$. The channel $P_{lin}$ is denoted by $P_{lin}=\displaystyle\sum_{k=1}^l p_k \mathcal{C}_{A_k}$.
\end{mydef}

The channel $P_{lin}$ can be seen as a box where we have a collection of matrices. At each channel use, a matrix $A_k$ from the box is chosen randomly according to the probabilities $p_k$, and the output of the channel is $A_k^T\vec{x}$, together with the index $k$ (so the receiver knows which matrix has been used).

\subsection{Characterizing non-losing channels}

We are interested in finding the channels whose capacity region is preserved by the polarization process.

\begin{myprop}
If $\{A_k,A_k':1\leq k\leq l\}$ is a set of matrices such that $\textrm{\emph{span}}(A_k)=\textrm{\emph{span}}(A_k')\;\forall k$, then the two channels $\displaystyle P=\sum_{k=1}^l p_k \mathcal{C}_{A_k}$ and $\displaystyle P'=\sum_{k=1}^l p_k \mathcal{C}_{A_k'}$ are equivalent.
\end{myprop}
\begin{proof}
If $\textrm{span}(A_k)=\textrm{span}(A_k')$, we can determine $A_k^T \vec{x}$ from ${A'_k}^T \vec{x}$ and vice versa. Therefore, from the output of $P$, we can deterministically obtain the output of $P'$ and vice versa. In this sense, $P$ and $P'$ are equivalent, and have the same capacity region.
\end{proof}

\begin{mynot}
Motivated by the above proposition, we will write $\displaystyle P\equiv\sum_{k=1}^l p_k\mathcal{C}_{V_k}$ (where $\{V_k\}_{1\leq k\leq l}$ is a set of $l$ subspaces of $\mathbb{F}_q^m$), whenever $P$ is equivalent to $\displaystyle \sum_{k=1}^l p_k\mathcal{C}_{A_k}$ and $\textrm{\emph{span}}(A_k)=V_k$.
\end{mynot}

\begin{myprop}
If $\displaystyle P\equiv\sum_{k=1}^l p_k\mathcal{C}_{V_k}$, then $\displaystyle I[S](P)=\sum_{k=1}^l p_k\emph{\textrm{dim}}\big(\emph{\textrm{proj}}_S(V_k)\big)$ for all $S\subset\{1,\ldots,m\}$. Where $\emph{\textrm{proj}}_S$ denotes the canonical projection on $\mathbb{F}_q^S$ defined by $\emph{\textrm{proj}}_S(\vec{x})=\emph{\textrm{proj}}_S(x_1,\ldots,x_m)= (x_{i_1},\ldots,x_{i_{|S|}})$ for $\vec{x}=(x_1,\ldots,x_m)\in \mathbb{F}_q^m$ and $S=\{i_1,\ldots,i_{|S|}\}$.
\end{myprop}
\begin{proof}
Let $X_1,\ldots,X_m$ be the input to the channel $\displaystyle\sum_{k=1}^l p_k\mathcal{C}_{A_k}$ (where $A_k$ spans $V_k$), and let $K,\vec{Y}$ be the output of it. We have:
\begin{align*}
H\big(X(S)|&K,\vec{Y},X(S^c)\big)\\
&=\sum_{k,\vec{y}} \textrm{P}_{K,\vec{Y}}(k,\vec{y})H\big(X(S)|k,\vec{y},X(S^c)\big)=\sum_{k,\vec{y}}\sum_{\vec{x}}\textrm{P}_{K,\vec{Y}|\vec{X}}(k,\vec{y}|\vec{x})
\textrm{P}_{\vec{X}}(\vec{x}) H\big(X(S)|k,\vec{y},X(S^c)\big)\\
&=\sum_{k,\vec{y}}\sum_{\substack{\vec{x},\\A_k^T\vec{x}=\vec{y}}}p_k\textrm{P}_{\vec{X}}(\vec{x}) H\big(X(S)|k,\vec{y},X(S^c)\big)=\sum_{k} p_k H\big(X(S)|A_k^T\vec{X},X(S^c)\big)\\
&=\sum_{k} p_k H\big(X(S)|A_k(S)^T\vec{X}(S),X(S^c)\big) = \sum_{k} p_k H\big(X(S)|A_k(S)^T\vec{X}(S)\big).
\end{align*}
The last equality follows from the fact that $X(S)$ and $X(S^c)$ are independent. $A_k(S)$ is obtained from $A_k$ by taking the rows corresponding to $S$. For a given value of $A_k(S)^T\vec{X}(S)$, we have $q^{d_k}$ possible values of $\vec{X}(S)$ with equal probabilities, where $d_k$ is the dimension of the null space of the mapping $\vec{X}(S)\rightarrow A_k(S)^T\vec{X}(S)$, so we have $H\big(X(S)|A_k(S)^T\vec{X}(S)\big)=d_k$.

On the other hand, $|S|-H\big(X(S)|A_k(S)^T\vec{X}(S)\big)=|S|-d_k$ is the dimension of the range space of the the mapping $\vec{X}(S)\rightarrow A_k(S)^T\vec{X}(S)$, which is also equal to the rank of $A_k(S)^T$. Therefore, we have:
\begin{align*}
|S|-H\big(X(S)|A_k(S)^T\vec{X}(S)\big)&=
\textrm{rank}(A_k(S)^T)=\textrm{rank}\big(A_k(S)\big)=\textrm{dim}\Big(\textrm{span} \big(A_k(S)\big)\Big)\\
&=\textrm{dim}\Big(\textrm{proj}_S\big( \textrm{span}(A_k)\big)\Big)=\textrm{dim}\big(\textrm{proj}_S(V_k)\big).
\end{align*}
We conclude:
\begin{align*}
I(X(S);K,Y,X(S^c))&=H(X(S))-H(X(S)|K,Y,X(S^c))=|S|-\sum_{k} p_k H(X(S)|A_k(S)^T\vec{X}(S))\\
&=\sum_{k} p_k \big(|S|-H(X(S)|A_k(S)^T\vec{X}(S))\big)=\sum_{k} p_k (|S|-d_k)\\
&=\sum_{k} p_k \textrm{dim}(\textrm{proj}_S(V_k)).
\end{align*}
\end{proof}

\begin{myprop}
If $\displaystyle P\equiv\sum_{k=1}^l p_k\mathcal{C}_{V_k}$ then:
\begin{itemize}
  \item $\displaystyle P^-\equiv\sum_{k_1=1}^l\sum_{k_2=1}^l p_{k_1}p_{k_2}\mathcal{C}_{V_{k_1}\cap V_{k_2}}$.
  \item $\displaystyle P^+\equiv\sum_{k_1=1}^l\sum_{k_2=1}^l p_{k_1}p_{k_2}\mathcal{C}_{V_{k_1} + V_{k_2}}$.
\end{itemize}
\end{myprop}
\begin{proof}
Suppose without lost of generality that $\displaystyle P=\sum_{k=1}^l p_k\mathcal{C}_{A_k}$ where $A_k$ spans $V_k$. Let $\vec{U}_1$ be an arbitrarily distributed random vector in $\mathbb{F}_q^m$ (not necessarily uniform), let $\vec{U}_2$ be a uniformly distributed random vector in $\mathbb{F}_q^m$ and independent of $\vec{U}_1$. Let $\vec{X}_1=\vec{U}_1+\vec{U}_2$ and $\vec{X}_2=\vec{U}_2$. Let $(K_1,A_{K_1}^T\vec{X}_1)$ and $(K_2,A_{K_2}^T\vec{X}_2)$ be the output of $P$ when the input is $X_1$ and $X_2$ respectively. Then the channel $\vec{U_1}\rightarrow(K_1,A_{K_1}^T\vec{X}_1,K_2,A_{K_2}^T\vec{X}_2)$ is equivalent to $P^-$ with $\vec{U_1}$ as input. We did not put any constraint on the distribution of $\vec{U}_1$ (such as saying that $\vec{U}_1$ is uniform) because in general, the model of a channel is characterized by its conditional probabilities and no assumption is made on the input probabilities.

Fix $K_1=k_1$ and $K_2=k_2$, let $A_{k_1\wedge k_2},B_{k_1}$ and $B_{k_2}$ be three matrices chosen such that $A_{k_1\wedge k_2}$ spans $V_{k_1}\cap V_{k_2}$, $A_{k_1}:=[A_{k_1\wedge k_2}\;B_{k_1}]$ spans $V_{k_1}$, $A_{k_2}:=[A_{k_1\wedge k_2}\;B_{k_2}]$ spans $V_{k_2}$, and the columns of $[A_{k_1\wedge k_2}\;B_{k_1}\;B_{k_2}]$ are linearly independent. Then knowing $A_{k_1}^T\vec{X}_1$ and $A_{k_2}^T\vec{X}_2$ is equivalent to knowing $A_{k_1\wedge k_2}^T(\vec{U}_1+\vec{U}_2)$, $B_{k_1}^T(\vec{U}_1+\vec{U}_2)$, $A_{k_1\wedge k_2}^T\vec{U}_2$ and $B_{k_2}^T\vec{U}_2$, which is equivalent to knowing $\vec{T}_{k_1,k_2}^1=A_{k_1\wedge k_2}^T\vec{U}_1$, $\vec{T}_{k_1,k_2}^2=B_{k_1}^T(\vec{U}_1+\vec{U}_2)$ and $\vec{T}_{k_1,k_2}^3=[A_{k_1\wedge k_2}\;B_{k_2}]^T\vec{U}_2$. We conclude that $P^-$ is equivalent to the channel:
$$\vec{U}_1\rightarrow \big(K_1,K_2,\vec{T}_{K_1,K_2}^1,\vec{T}_{K_1,K_2}^2,\vec{T}_{K_1,K_2}^3\big).$$
Conditioned on $(K_1,K_2,\vec{T}_{K_1,K_2}^1)$ we have $[B_{K_1}\;A_{K_1\wedge K_2}\;B_{K_2}]^T\vec{U}_2$ is uniform (since the matrix $[B_{K_1}\;A_{K_1\wedge K_2}\;B_{K_2}]$ is full rank) and independent of $\vec{U}_1$, so $[A_{K_1\wedge K_2}\;B_{K_2}]^T\vec{U}_2$ is independent of $(B_{K_1}^T\vec{U}_2,\vec{U}_1)$, which implies that $[A_{K_1\wedge K_2}\;B_{K_2}]^T\vec{U}_2$ is independent of $\big(B_{K_1}^T(\vec{U}_1+\vec{U}_2),\vec{U}_1\big)$. Also conditioned on $(K_1,K_2,\vec{T}_{K_1,K_2}^1)$, $B_{K_1}^T\vec{U}_2$ is uniform and independent of $\vec{U}_1$, which implies that $\vec{U}_1$ is independent of $B_{K_1}^T(\vec{U}_1+\vec{U}_2)$, and this is because the columns of $B_{K_1}$ and $A_{K_1\wedge K_2}$ are linearly independent. We conclude that conditioned on $(K_1,K_2,\vec{T}_{K_1,K_2}^1)$, $\vec{U}_1$ is independent of $\big(\vec{T}_{K_1,K_2}^2,\vec{T}_{K_1,K_2}^3\big)$. Therefore, $\big(K_1,K_2,\vec{T}_{K_1,K_2}^1\big) =\big(K_1,K_2,A_{k_1\wedge k_2}^T\vec{U}_1\big)$ form sufficient statistics. We conclude that $P^-$ is equivalent to the channel:

$$\vec{U}_1\rightarrow =\big(K_1,K_2,A_{k_1\wedge k_2}^T\vec{U}_1\big).$$

And since $\textrm{P}(K_1=k_1,K_2=k_2)=p_{k_1}p_{k_2}$, and $A_{k_1\wedge k_2}$ spans $V_{k_1}\cap V_{k_2}$ we conclude that $\displaystyle P^-\equiv\sum_{k_1=1}^l \sum_{k_2=1}^l p_{k_1}p_{k_2}\mathcal{C}_{V_{k_1}\cap V_{k_2}}$.

Now let $\vec{U}_2$ be arbitrarily distributed in $\mathbb{F}_q^m$ (not necessarily uniform) and $\vec{U}_1$ be a uniformly distributed random vector in $\mathbb{F}_q^m$ independent of $\vec{U}_2$. Let $\vec{X}_1=\vec{U}_1+\vec{U}_2$ and $\vec{X}_2=\vec{U}_2$. Let $(K_1,A_{K_1}^T\vec{X}_1)$ and $(K_2,A_{K_2}^T\vec{X}_2)$ be the output of $P$ when the input is $X_1$ and $X_2$ respectively. Then the channel $\vec{U_2}\rightarrow(K_1,A_{K_1}^T\vec{X}_1,K_2,A_{K_2}^T\vec{X}_2,\vec{U_1})$ is equivalent to $P^+$ with $\vec{U_2}$ as input. Note that the uniform distribution constraint is now on $\vec{U}_1$ and no constraint is put on the distribution of $\vec{U}_2$, since now $\vec{U}_2$ is the input to the channel $P^+$.

Knowing $A_{K_1}^T\vec{X}_1$, $A_{K_2}^T\vec{X}_2$ and $\vec{U_1}$ is equivalent to knowing $A_{K_1}^T(\vec{U}_1+\vec{U}_2)$, $A_{K_2}^T\vec{U}_2$ and $\vec{U_1}$, which is equivalent to knowing $A_{K_1}^T\vec{U}_2$, $A_{K_2}^T\vec{U}_2$ and $\vec{U_1}$. So $P^+$ is equivalent to the channel:

$$\vec{U_2}\rightarrow \big(K_1,K_2,[A_{k_1}\;A_{k_2}]^T\vec{U}_2,\vec{U}_1\big).$$

And since $\vec{U}_1$ is independent of $\vec{U}_2$, the above channel (and hence $P+$) is equivalent to the channel:

$$\vec{U_2}\rightarrow \big(K_1,K_2,[A_{k_1}\;A_{k_2}]^T\vec{U}_2\big).$$

We also have $\textrm{P}(K_1=k_1,K_2=k_2)=p_{k_1}p_{k_2}$, and $[A_{k_1}\;A_{k_2}]$ spans $V_{k_1}+V_{k_2}$. We conclude that $\displaystyle P^+\equiv\sum_{k_1=1}^l \sum_{k_2=1}^l p_{k_1}p_{k_2}\mathcal{C}_{V_{k_1} + V_{k_2}}$.
\end{proof}

\begin{mylem}
Let $\displaystyle P\equiv\sum_{k=1}^l p_k\mathcal{C}_{V_k}$ and $S\subset\{1,\ldots,m\}$, then
\begin{align*}
&\frac{1}{2}\big(I[S](P^-)+I[S](P^+)\big)=I[S](P) \; \Leftrightarrow \; \Big(\forall (k_1,k_2);\;\emph{\textrm{proj}}_S(V_{k_1}\cap V_{k_2})=\emph{\textrm{proj}}_S(V_{k_1})\cap\emph{\textrm{proj}}_S(V_{k_2})\Big).
\end{align*}
\end{mylem}
\begin{proof}
We know that if $V$ and $V'$ are two subspaces of $\mathbb{F}_q^m$, then $\textrm{proj}_S(V\cap V')\subset\textrm{proj}_S(V)\cap\textrm{proj}_S(V')$ and $\textrm{proj}_S(V+V')=\textrm{proj}_S(V)+\textrm{proj}_S(V')$, which implies that:
\begin{itemize}
\item $\textrm{dim}\big(\textrm{proj}_S(V\cap V')\big)\leq\textrm{dim}\big(\textrm{proj}_S(V)\cap\textrm{proj}_S(V')\big)$.
\item $\textrm{dim}\big(\textrm{proj}_S(V+ V')\big)=\textrm{dim}\big(\textrm{proj}_S(V)+\textrm{proj}_S(V')\big)$.
\end{itemize}
We conclude:
\vspace*{-2mm}
\begin{align*}
\textrm{dim}&\big(\textrm{proj}_S (V\cap V')\big) + \textrm{dim}\big(\textrm{proj}_S(V+ V')\big)\\
&\leq \textrm{dim}\big(\textrm{proj}_S(V)\cap\textrm{proj}_S(V')\big)+\textrm{dim} \big(\textrm{proj}_S(V)+\textrm{proj}_S(V')\big)=\textrm{dim}\big(\textrm{proj}_S(V)\big) + \textrm{dim}\big(\textrm{proj}_S(V')\big).
\end{align*}
Therefore:
\begin{align}
\frac{1}{2}&\big(I[S](P^-)+I[S](P^+)\big)\notag\\
&= \frac{1}{2}\sum_{k_1=1}^l\sum_{k_2=1}^l p_{k_1}p_{k_2}\textrm{dim}\big(\textrm{proj}_S(V_{k_1}\cap V_{k_2})\big) + \frac{1}{2} \sum_{k_1=1}^l\sum_{k_2=1}^l \notag p_{k_1}p_{k_2}\textrm{dim}\big(\textrm{proj}_S(V_{k_1}+ V_{k_2})\big)\Big)\\
&=\frac{1}{2}\bigg(\sum_{k_1=1}^l\sum_{k_2=1}^l p_{k_1}p_{k_2}\Big(\textrm{dim}\big(\textrm{proj}_S(V_{k_1}\cap V_{k_2})\big)+\textrm{dim}\big(\textrm{proj}_S(V_{k_1}+ V_{k_2})\big)\Big)\bigg)\notag\\
\label{ineqb}&\leq \frac{1}{2}\bigg(\sum_{k_1=1}^l\sum_{k_2=1}^l p_{k_1}p_{k_2}\Big(\textrm{dim}\big(\textrm{proj}_S( V_{k_1})\big)+\textrm{dim}\big(\textrm{proj}_S(V_{k_2})\big)\Big)\bigg) \\
& =  \frac{1}{2}\Big(\sum_{k_1=1}^l p_{k_1}\textrm{dim}\big(\textrm{proj}_S( V_{k_1})\big)+ \sum_{k_2=1}^l p_{k_2}\textrm{dim}\big(\textrm{proj}_S(V_{k_2})\big)\Big)= \frac{1}{2}(I[S](P)+I[S](P))=I[S](P).\notag
\end{align}

Thus, if we have $\textrm{proj}_S(V_{k_1}\cap V_{k_2})\subsetneq\textrm{proj}_S(V_{k_1})\cap\textrm{proj}_S(V_{k_2})$ for some $k_1,k_2$, then we have $\textrm{dim}\big(\textrm{proj}_S(V_{k_1}\cap V_{k_2})\big)<\textrm{dim}\big(\textrm{proj}_S(V_{k_1})\cap \textrm{proj}_S(V_{k_2})\big)$, and the inequality \ref{ineqb} is strict. We conclude that:
\begin{align*}
&\frac{1}{2}\big(I[S](P^-)+I[S](P^+)\big)=I[S](P) \; \Leftrightarrow \; \Big(\forall (k_1,k_2),\;\textrm{proj}_S(V_{k_1}\cap V_{k_2})=\textrm{proj}_S(V_{k_1})\cap\textrm{proj}_S(V_{k_2})\Big).
\end{align*}
\end{proof}

\begin{mydef}
Let $\mathcal{V}$ be a set of subspaces of $\mathbb{F}_q^m$, we define the closure of $\mathcal{V}$, $cl(\mathcal{V})$, as being the minimal set of subspaces of $\mathbb{F}_q^m$ closed under the two operations $\cap$ and $+$, and including $\mathcal{V}$. We say that the set $\mathcal{V}$ is consistent with respect to $S\subset\{1,\ldots,m\}$ if and only if it satisfies the following property:
\begin{align*}
\forall V_1,V_2&\in cl(\mathcal{V});\;\emph{proj}_S(V_1\cap V_2)=\emph{proj}_S(V_1)\cap\emph{proj}_S(V_2).
\end{align*}
\end{mydef}

\begin{mycor}
If $\mathcal{V}=\{V_k:1\leq k\leq l\}$. $I[S](P)$ is preserved by the polarization process if and only if $\mathcal{V}$ is consistent with respect to $S$.
\end{mycor}
\begin{proof}
During the polarization process, we are performing successively the $\cap$ and $+$ operators, which means that we'll reach the closure of $\mathcal{V}$ after a finite number of steps. So $I[S](P)$ is preserved if and only if the above lemma applies to $cl(\mathcal{V})$.
\end{proof}

The above corollary gives a characterization for a combination of linear channels to preserve $I[S](P)$. However, this characterization involves using the closure operator. The next proposition gives a sufficient condition that uses only the initial configuration of subspaces $\mathcal{V}$. This proposition gives a certain ``geometric'' view of what the subspaces should look like if we don't want to lose.

\begin{myprop}
\label{proptop}
Suppose there exists a subspace $V_S$ of dimension $|S|$ such that $\emph{proj}_S(V_S)=\mathbb{F}_q^S$, and suppose that for every $V\in\mathcal{V}$ we have $\emph{proj}_S(V_S\cap V)=\emph{proj}_S(V)$, then $I[S](P)$ is preserved by the polarization process.
\end{myprop}
\begin{proof}
Let $V_S$ be a subspace satisfying the hypothesis, then it satisfies also the hypothesis if we replace $\mathcal{V}$ by it's closure: If $V_1$ and $V_2$ are two arbitrary subspaces satisfying $$\textrm{proj}_S(V_S\cap V_1)=\textrm{proj}_S(V_1)\;\textrm{and}\;\textrm{proj}_S(V_S\cap V_2)=\textrm{proj}_S(V_2),$$ then $\textrm{proj}_S(V_1)\subset \textrm{proj}_S\big(V_S\cap(V_1+V_2)\big)$ and $\textrm{proj}_S(V_2)\subset \textrm{proj}_S\big(V_S\cap(V_1+V_2)\big)$, which implies $\textrm{proj}_S(V_1+V_2)=\textrm{proj}_S(V_1)+\textrm{proj}_S(V_2)\subset \textrm{proj}_S\big(V_S\cap(V_1+V_2)\big)$. Therefore, $\textrm{proj}_S\big(V_S\cap(V_1+V_2)\big)=\textrm{proj}_S(V_1+V_2)$ since the inverse inclusion is trivial.

Now let $\vec{x}\in\textrm{proj}_S(V_1) \cap \textrm{proj}_S(V_2)$, then $\vec{x}\in\textrm{proj}_S(V_1)=\textrm{proj}_S(V_1\cap V_S)$ and similarly $\vec{x}\in\textrm{proj}_S(V_2\cap V_S)$ which implies that there are two vectors $\vec{x_1}\in V_1\cap V_S$ and $\vec{x_2}\in V_2\cap V_S$ such that $\vec{x}=\textrm{proj}_S(\vec{x}_1)=\textrm{proj}_S(\vec{x_2})$. And since $\textrm{proj}_S(V_S)=\mathbb{F}_q^S$ and $\textrm{dim}(V_S)=|S|$, then the mapping $\textrm{proj}_S:V_S\rightarrow\mathbb{F}_q^S$ is invertible and so $\vec{x}_1=\vec{x}_2$ which implies that $\vec{x}\in\textrm{proj}_S(V_1\cap V_2\cap V_S)$. Thus $\textrm{proj}_S(V_1) \cap \textrm{proj}_S(V_2)\subset\textrm{proj}_S(V_1\cap V_2)\subset\textrm{proj}_S(V_1\cap V_2\cap V_S)$. We conclude that $\textrm{proj}_S(V_1) \cap \textrm{proj}_S(V_2)=\textrm{proj}_S(V_1\cap V_2)=\textrm{proj}_S(V_1\cap V_2\cap V_S)$ since the inverse inclusions are trivial.

We conclude that the set of subspaces $V$ satisfying $\textrm{proj}_S(V\cap V_S)=\textrm{proj}_S(V)$ is closed under the two operators $\cap$ and $+$. And since $\mathcal{V}$ is a subset of this set, $cl(\mathcal{V})$ is a subset as well. Now let $V_1,V_2\in cl(\mathcal{V})$, then $\textrm{proj}_S(V_S\cap V_1)=\textrm{proj}_S(V_1)$ and $\textrm{proj}_S(V_S\cap V_2)=\textrm{proj}_S(V_2)$. Then $\textrm{proj}_S(V_1) \cap \textrm{proj}_S(V_2)=\textrm{proj}_S(V_1\cap V_2)$ as we have seen in the previous paragraph. We conclude that $\mathcal{V}$ is consistent with respect to $S$ and so $I[S](P)$ is preserved.
\end{proof}

\begin{myconj}
The condition in proposition \ref{proptop} is necessary: If $I[S](P)$ is preserved by the polarization process, there must exist a subspace $V_S$ of dimension $|S|$ such that $\emph{proj}_S(V_S)=\mathbb{F}_q^S$, and for every $V\in\mathcal{V}$ we have $\emph{proj}_S(V_S\cap V)=\emph{proj}_S(V)$.
\end{myconj}

\subsection{Maximal loss in the dominant face}

After characterizing the non-losing channels, we are now interested in studying the amount of loss in the capacity region. In order to simplify the problem, we only study it in the case of binary input 2-user MAC since the $q$-ary case is similar.

Since we only have 5 subspaces of $\mathbb{F}_2^2$, we write $\displaystyle P\equiv \sum_{k=0}^4 p_k\mathcal{C}_{V_k}$ (here $p_k$ are allowed to be zero), where $V_0$, \ldots, $V_4$ are the 5 possible subspaces of $\mathbb{F}_2^2$:
\begin{itemize}
\item $V_0=\{(0,0)\}$.
\item $V_1=\{(0,0),(1,0)\}$.
\item $V_2=\{(0,0),(0,1)\}$.
\item $V_3=\{(0,0),(1,1)\}$.
\item $V_4=\{(0,0),(1,0),(0,1),(1,1)\}$.
\end{itemize}

We have $I[\{1\}](P)=p_1+p_3+p_4$, $I[\{2\}](P)=p_2+p_3+p_4$ and $I(P)=I[\{1,2\}](P)=p_1+p_2+p_3+2p_4$.

\begin{mydef}
Let $\displaystyle P\equiv\sum_{k=0}^4 p_k\mathcal{C}_{V_k}$ and $s\in\{-,+\}^n$, we write $p_k^s$ to denote the component of $V_k$ in $P^s$, i.e. we have $\displaystyle P^s\equiv\sum_{k=0}^4 p_k^s\mathcal{C}_{V_k}$.

We denote the average of $p_k^s$ on all possible $s\in\{-,+\}^n$ by $p_k^{(n)}$. i.e. $\displaystyle p_k^{(n)}=\frac{1}{2^n}\sum_{s\in\{-,+\}^n}p_k^s$. $p_k^{(\infty)}$ is the limit of $p_k^{(n)}$ as $n$ tends to infinity. We will see later that $p_k^{(n)}$ is increasing if $k\in\{0,4\}$ and decreasing if $k\in\{1,2,3\}$. This shows that the limit of $p_k^{(n)}$ as $n$ tends to infinity always exists, and $p_k^{(\infty)}$ is well defined.

We denote the average of $I[\{1\}](P^s)$ (resp. $I[\{2\}](P^s)$ and $I(P^s)$) on all possible $s\in\{-,+\}^n$ by $I_1^{(n)}$ (resp. $I_2^{(n)}$ and $I^{(n)}$). We have $I_1^{(n)}=p_1^{(n)}+p_3^{(n)}+p_4^{(n)}$, $I_2^{(n)}=p_2^{(n)}+p_3^{(n)}+p_4^{(n)}$ and $I^{(n)}=p_1^{(n)}+p_2^{(n)}+p_3^{(n)}+2p_4^{(n)}$. If $n$ tends to infinity we get $I_1^{(\infty)}=p_1^{(\infty)}+p_3^{(\infty)}+p_4^{(\infty)}$, $I_2^{(\infty)}=p_2^{(\infty)}+p_3^{(\infty)}+p_4^{(\infty)}$ and $I^{(\infty)}=p_1^{(\infty)}+p_2^{(\infty)}+p_3^{(\infty)}+2p_4^{(\infty)}$.
\end{mydef}

\begin{mydef}
We say that we have maximal loss in the dominant face in the polarization process, if the dominant face of the capacity region converges to a single point.
\end{mydef}

\begin{myrem}
The symmetric capacity region after $n$ polarization steps is the average of the symmetric capacity regions of all the channels $P^s$ obtained after $n$ polarization steps ($s\in\{-,+\}^n$). Therefore, this capacity region is given by:
\begin{align*}
\mathcal{J}(P^{(n)}):=\Big\{(R_1,& R_2):\;0\leq R_1\leq I_1^{(n)},\;0\leq R_2\leq I_2^{(n)},\;0\leq R_1+R_2\leq I^{(n)}\Big\}.
\end{align*}
The above capacity region converges to the ``final capacity region'':
\begin{align*}
\mathcal{J}(P^{(\infty)}):=\Big\{(R_1,& R_2):\;0\leq R_1\leq I_1^{(\infty)},\;0\leq R_2\leq I_2^{(\infty)},\;0\leq R_1+R_2\leq I^{(\infty)}\Big\}.
\end{align*}
The dominant face converges to a single point if and only if $I^{(\infty)}=I_1^{(\infty)}+I_2^{(\infty)}$, which is equivalent to $p_1^{(\infty)}+p_2^{(\infty)}+ p_3^{(\infty)}+2p_4^{(\infty)}=p_1^{(\infty)}+p_2^{(\infty)}+2p_3^{(\infty)} +2p_4^{(\infty)}$. We conclude that we have maximal loss in the dominant face if and only if $p_3^{(\infty)}=0$.
\label{remtem}
\end{myrem}

\begin{mylem}
The order of $p_1,p_2$ and $p_3$ remains the same by the polarization process. e.g. if $p_1<p_3<p_2$ then $p_1^s<p_3^s<p_2^s$, and if $p_2=p_3<p_1$ then $p_2^s=p_3^s<p_1^s$ for all $s\in\{-,+\}^n$.
\end{mylem}
\begin{proof}
We have $\displaystyle P^-=\sum_{k=0}^4\sum_{k'=0}^4 p_kp_{k'}\mathcal{C}_{V_k\cap V_{k'}}$ and $\displaystyle P^+=\sum_{k=0}^4\sum_{k'=0}^4 p_kp_{k'}\mathcal{C}_{V_k+V_{k'}}$. Therefore, we have:
\begin{align*}
p_0^-&=p_0^2+2p_0(p_1+p_2+p_3+p_4)+2(p_1p_2+p_2p_3+p_1p_3),\\
p_1^-&=p_1^2+2p_1p_4,\\
p_2^-&=p_2^2+2p_2p_4,\\
p_3^-&=p_3^2+2p_3p_4,\\
p_4^-&=p_4^2,\\
\\
p_0^+&=p_0^2,\\
p_1^+&=p_1^2+2p_1p_0,\\
p_2^+&=p_2^2+2p_2p_0,\\
p_3^+&=p_3^2+2p_3p_0,\\
p_4^+&=p_4^2+2p_4(p_1+p_2+p_3+p_4)+2(p_1p_2+p_2p_3+p_1p_3).
\end{align*}

We can easily see that the order of $p_1^-$,$p_2^-$ and $p_3^-$ is the same as that of $p_1,p_2$ and $p_3$. This is also true for $p_1^+$,$p_2^+$ and $p_3^+$. By using a simple induction on $s$, we conclude that the order of $p_1^s,p_2^s$ and $p_3^s$ is the same as that of $p_1,p_2$ and $p_3$ for all $s\in\{-,+\}^n$.
\end{proof}

\begin{myrem}
The equations that give $\{p_k^-\}_{0\leq k\leq 4}$ and $\{p_k^+\}_{0\leq k\leq 4}$ in terms of $\{p_k\}_{0\leq k\leq 4}$ clearly show that $p_k^{(n)}$ is increasing if $k\in\{0,4\}$ and decreasing if $k\in\{1,2,3\}$.
\end{myrem}

\begin{mylem}
For $k\in\{1,2,3\}$, if $\exists k'\in\{1,2,3\}\setminus\{k\}$ such that $p_k\leq p_{k'}$ then $$p_k^{(\infty)}=\lim_{l\rightarrow\infty}\frac{1}{2^n}\sum_{s\in\{-,+\}^n}p_k^s=0.$$ In other words, the component of $V_k$ is killed by that of $V_{k'}$.
\end{mylem}
\begin{proof}
We know from theorem \ref{mainthemac} that the channel $P^s$ converges almost surely to a deterministic linear channel as $n$ tends to infinity (we treat $s$ as being a uniform random variable in $\{-,+\}^n$). Therefore, the vector $(p_0^s,p_1^s,p_2^s,p_3^s,p_4^s)$ converges almost surely to one of the following vectors: $(1,0,0,0,0)$, $(0,1,0,0,0)$, $(0,0,1,0,0)$, $(0,0,0,1,0)$ or $(0,0,0,0,1)$. In particular, $p_k^s$ converges almost surely to 0 or 1.

Since $p_k\leq p_{k'}$ then $p_k^s\leq p_{k'}^s$ for any $s$, and so $p_k^s$ cannot converge to 1 because otherwise the limit of $p_{k'}^s$ would also be equal to 1, which is not possible since none of the 5 possible vectors contain two ones. We conclude that $p_k^s$ converges almost surely to 0, which means that $p_k^{(n)}$ (the average of $p_k^s$ on all possible $s\in\{-,+\}^n$) converges to 0. Therefore, $p_k^{(\infty)}=0$.
\end{proof}

\begin{myprop}
If $p_3\leq\max\{p_1,p_2\}$, then we have maximal loss in the dominant face.
\end{myprop}
\begin{proof}
If $p_3\leq\max\{p_1,p_2\}$, then by the previous lemma we have $p_3^{(\infty)}=0$. Therefore, we have maximal loss in the dominant face (see remark \ref{remtem}).
\end{proof}

\begin{mycor}
If we do not have maximal loss in the dominant face then the final capacity region (to which the capacity region is converging) must be symmetric.
\end{mycor}
\begin{proof}
From the above proposition we conclude that $p_3>\max\{p_1,p_2\}$ and from lemma 9 we conclude that $p_1^{(\infty)}=p_2^{(\infty)}=0$. Thus, $I_1^{(\infty)}=I_2^{(\infty)}=p_3^{(\infty)}+p_4^{(\infty)}$ and the final capacity region is symmetric.  In particular, it contains the ``equal-rates'' rate vector.
\end{proof}

\begin{myconj}
The condition in \emph{proposition 9} is necessary for having maximal loss in the dominant face. i.e. if $p_3>\max\{p_1,p_2\}$, then we do not have maximal loss in the dominant face.
\end{myconj}

\section{Conclusion}

We have shown that quasigroup is a sufficient property for an operation to ensure polarization when it is used in an Ar{\i}kan-like construction. The determination of a more general property that is both necessary and sufficient  remains an open problem.

In the case of MACs, we have shown that while the symmetric sum capacity is achievable by polar codes, we may lose some rate vectors from the capacity region by polarization. We have studied this loss in the case where the channel is a combination of linear channels, and we derived a characterization of non-losing channels in this special case. We have also derived a sufficient condition for having maximal loss in the dominant face in the capacity region in the case of binary input 2-user MAC.

It is possible to achieve the whole capacity region of any MAC by applying time sharing of polar codes. An important question, which remains open, is whether it is possible to find a coding scheme, based only on an Arikan-like construction, which achieves the whole symmetric capacity region.

\bibliographystyle{IEEEtran}
\bibliography{bibliofile}
\end{document}